\documentclass[twocolumn,english,jmp,amsmath,amssymb,reprint,nofootinbib,twocolumn,altaffilletter,superscriptaddress]{revtex4-1}
\usepackage[T1]{fontenc}
\usepackage[latin9]{inputenc}
\usepackage{color}
\usepackage{amsmath}
\usepackage{amsthm}
\usepackage{amssymb}

\makeatletter
\theoremstyle{plain}
\newtheorem{thm}{\protect\theoremname}
\theoremstyle{plain}
\newtheorem{lem}[thm]{\protect\lemmaname}
\theoremstyle{plain}
\newtheorem{prop}[thm]{\protect\propositionname}
\theoremstyle{remark}
\newtheorem{rem}[thm]{\protect\remarkname}
\theoremstyle{definition}
\newtheorem{defn}[thm]{\protect\definitionname}
\theoremstyle{plain}
\newtheorem{cor}[thm]{\protect\corollaryname}

\usepackage[table]{xcolor}
\usepackage{colortbl}

\usepackage{tikz}
\usetikzlibrary{tikzmark}
\usepackage{dsfont}

\usepackage{titlesec}
\titleformat*{\paragraph}{\bf\small}

\renewcommand{\Re}{\operatorname{Re}}
\renewcommand{\Im}{\operatorname{Im}}
\usepackage{array}
\usepackage{booktabs}

\usepackage{babel}
\providecommand{\corollaryname}{Corollary}
  \providecommand{\definitionname}{Definition}
  \providecommand{\lemmaname}{Lemma}
  \providecommand{\propositionname}{Proposition}
  \providecommand{\remarkname}{Remark}
\providecommand{\theoremname}{Theorem}

\usepackage{babel}
\providecommand{\corollaryname}{Corollary}
  \providecommand{\definitionname}{Definition}
  \providecommand{\lemmaname}{Lemma}
  \providecommand{\propositionname}{Proposition}
  \providecommand{\remarkname}{Remark}
\providecommand{\theoremname}{Theorem}

\makeatother

\usepackage{babel}
\providecommand{\corollaryname}{Corollary}
\providecommand{\definitionname}{Definition}
\providecommand{\lemmaname}{Lemma}
\providecommand{\propositionname}{Proposition}
\providecommand{\remarkname}{Remark}
\providecommand{\theoremname}{Theorem}

\begin{document}

\title{Inequivalent coherent state representations in group field theory}

\author{Alexander Kegeles}
\email{kegeles@aei.mpg.de}

\affiliation{Max Planck Institute for Gravitational Physics (Albert Einstein Institute),
~\\
Am Mühlenberg 1, 14476 Potsdam-Golm, Germany, EU}

\affiliation{Institute of Physics and Astronomy, University of Potsdam, ~\\
Karl-Liebknecht-Str. 24/25, 14476 Potsdam, Germany }

\author{Daniele Oriti}
\email{oriti@aei.mpg.de}

\affiliation{Max Planck Institute for Gravitational Physics (Albert Einstein Institute),
~\\
Am Mühlenberg 1, 14476 Potsdam-Golm, Germany, EU}

\author{Casey Tomlin}
\email{caseytomlin@gmail.com}

\affiliation{Max Planck Institute for Gravitational Physics (Albert Einstein Institute),
~\\
Am Mühlenberg 1, 14476 Potsdam-Golm, Germany, EU}

\affiliation{Booz Allen Hamilton ~~\\
 901 15th St. NW Washington, DC 20005, USA}
\begin{abstract}
In this paper we propose an algebraic formulation of group field theory
and consider non-Fock representations based on coherent states. We
show that we can construct representations with infinite number of
degrees of freedom on compact base manifolds. We also show that these
representations break translation symmetry. Since such representations
can be regarded as quantum gravitational systems with an infinite
number of fundamental pre-geometric building blocks, they may be more
suitable for the description of effective geometrical phases of the
theory. 
\end{abstract}
\maketitle

\section*{Introduction}

Many contemporary approaches to quantum gravity see spacetime and
geometry as collective phenomena of more fundamental degrees of freedom.
In such theories, a transition from fundamental and non-geometric
to the effective geometric level is often associated with a phase
transition and requires control over many degrees of freedom. A key
goal is then to provide a consistent description of this phase transition.
 In the algebraic formulation of quantum field theory, different
phases are associated with inequivalent representations of the operator
algebra of observables; the study of phase transitions becomes the
study of the operator algebra and its inequivalent representations.
In this paper we suggest an algebraic formulation of group field theory
(GFT), investigate its operator algebra and provide examples of its
inequivalent representations on a compact base manifold.

Group field theory \cite{Oriti:2006se,Oriti:2011jm,Krajewski:2012aw}
is one candidate theory that aims at the description of emergence
of geometry. It is a statistical quantum field theory in which space-time
geometry and dynamics of general relativity suppose to arise as an
effective field theory. It is closely related to canonical loop quantum
gravity (LQG) \cite{Thiemann:2002nj,Rovelli2008,Ashtekar:2017iip,Bodendorfer:2016uat}
and its covariant formulation in terms of spin foams \cite{Baez:1999sr,Perez:2012wv};
for details on this relation, see \cite{Oriti:2014uga}. On the other
hand, it can also be seen as a group-theoretic enrichment of random
tensor models \cite{Gurau:2011xp,Gurau:2016cjo,Rivasseau:2016wvy},
in which tensor indices over finite sets are replaced by field arguments
\cite{Rivasseau:2011hm,Rivasseau:2012yp,Rivasseau:2013uca,Rivasseau:2016wvy,Oriti:2014qoa}.

The quanta of GFT models formally describe point particles labeled
by a (generally non-abelian) Lie group in the same way that quanta
of ordinary field theories are formally labeled by points of spacetime.
However, canonical quantization and the resulting Hamiltonian dynamics
or evolution which entirely relies on a time variable cannot be applied
here since time does not (yet) exist. Still, a Hilbert space for ``particles
on the group'' can be defined guided by a discrete geometric intuition;
in particular, the GFT quanta can be understood as quantized simplices
(tetrahedra in 4 dimension), whose quantum algebra and single particle
Hilbert space are obtained by geometric quantization of a classical
discrete geometry (see for example \cite{Baez:1999tk,Conrady:2009px,Bianchi:2010gc}).
Applying second quantization techniques, one can construct a Fock
space of quantum simplices that serves as the Hilbert space for GFT.
The simplicial building blocks that are populating the Fock space
admit a dual interpretation in terms of spin network vertices \cite{Conrady:2009px,Bianchi:2010gc,Baratin:2010nn}.

Nevertheless, the Fock vacuum provides trivial topology and geometry
and therefore, can be intuitively considered ``far away'' from any
state that carries information about non-trivial smooth spacetime
geometry. On the other hand, finitely many-particle states in GFT
have a discrete geometric interpretation, shared with loop quantum
gravity and simplicial quantum gravity, and provide a notion of generalized
piecewise-flat geometries \cite{Oriti:2014uga}. However, for a description
of smooth geometries the number of degrees of freedom, or GFT quanta,
should be very large and states with an infinite particle number are
likely to be needed.

In turn, the interactions among large numbers of GFT quanta, i.e.
their collective behavior, may give rise to phase transitions, as
in any other non-trivial quantum field theory (see for example \cite{sewell2014quantum}).
New questions, then, arise: which phase of a given GFT model, if any,
admits a geometric interpretation and a description in terms of effective
field theory and general relativity? Which quantum representation
of the fundamental GFT is appropriate to the description of such geometric
physics? 

This prompts us to study the definition of new representations in
GFT, taking full advantage of its field-theoretic structures, and
complementing parallel work on GFT renormalization \cite{Benedetti:2014qsa,Geloun:2016qyb,Geloun:2016bhh,Benedetti:2015yaa,Carrozza:2016tih,Carrozza:2016vsq,Carrozza:2016tih,Carrozza:2017vkz}.
Our approach provides a GFT counterpart of similar studies, with identical
motivations, carried out in the context of canonical loop quantum
gravity, spin foam models, tensor models and dynamical triangulations
\cite{Koslowski:2011vn,Dittrich:2014wpa,Bahr:2015bra,Delcamp:2016dqo,Dittrich:2014mxa,Dittrich:2014ala,Ambjorn:2014gsa,Delepouve:2015nia,Bonzom:2016dwy,Ambjorn:2017tnl}.

Our work is motivated by the use of GFT coherent states in the extraction
of an effective continuum dynamics \cite{Gielen:2013kla,Oriti:2015qva,Oriti:2015rwa,Oriti:2016qtz,Oriti:2016ueo,Pithis:2016wzf,Gielen:2016dss},
and the requirement of an infinite number of degrees of freedom that
is needed for description of smooth geometries. To study these two
requirements we construct coherent state representations with an infinite
number of GFT quanta and study their relation with the Fock representation.
 The idea is to avoid the limiting procedures of the particle number
for thermodynamical potentials but instead define directly representations
that correspond to an infinite system. 

Using this approach we can explicitly formulate the theory on Hilbert
spaces with infinite particle number. Such Hilbert spaces could be
better suitable for a description of geometrical states. The structure
of the constructed representations is however still very simple and
more realistic representations with richer structure have to be understood
in future work.

In the first part of this paper (\ref{sec:Group-Field-Theory}) we
set up the algebraic formulation of GFT. Using this formulation in
the second part (\ref{sec:States-and-representations}) we show how
one can construct inequivalent representations for GFT and provide
simple examples of representations associated to infinite systems
with breaking of translation symmetry.

\subsection*{Notation}

In this paper we will use the following notation and conventions.
The base manifold of GFT is considered to be $G^{\times n}$ with
$G=SU\left(2\right)$ and some fixed $n\in\mathbb{N}$; it will be
denoted, $M\doteq G^{\times n}$. A generalization of statements from
this paper to compact Lie groups other than $SU\left(2\right)$ is
straightforward, but a treatment of non-compact base manifold requires
more care. Throughout the whole paper the letter $h$ is reserved
as an element of $G$, and $\text{d}h$ refers to the Haar measure
on $G$, the Haar integral on $G$ is denoted by $\int_{G}\,\left(\cdot\right)\,\text{d}h$.
The Haar measure is normalized to $1$, $\int_{G}\,\text{d}h=1$,
and is invariant under left and right multiplication and inversion
on $G$, that is for an integrable function $f$ and $h_{1},h_{2}\in G$
\begin{align}
\int_{G}\,f\left(h_{1}h\,h_{2}\right)\text{d}h & =\int_{G}\,f\left(h\right)\,\text{d}h.\\
\int_{G}\,f\left(h^{-1}\right)\,\text{d}h & =\int_{G}\,f\left(h\right)\,\text{d}h
\end{align}
It is a unique measure on $G$ with this properties. 

The letters $x$ and $y$ are reserved for elements of $M$, and $\text{d}x$
refers to the Haar measure on $M$, the Haar integral on $M$ is denoted
by $\int_{M}\,\left(\cdot\right)\,\text{d}x$. The Haar measure $\text{d}x$
is, as above, normalized to $1$ and invariant under left and right
multiplication as well as inversion on $M$. Whenever necessary, we
will use subscripts for the components of $x$ and write $x=\left(x_{1},x_{2},\cdots,x_{n}\right)\in M$.
We denote the Lie algebra of $M$ by $\mathfrak{m}$ and by convention
choose it to be isomorphic to the space of right invariant vector
fields on $M$.

We denote the space of square integrable functions on $M$ by $L^{2}\left(M,\text{d}x\right)$
and define the bracket $\left(\cdot,\cdot\right)_{L^{2}}$, such that
for any $f,g\in L^{2}\left(M,\text{d}x\right)$, 
\begin{equation}
\left(f,g\right)_{L^{2}}=\int_{M}\,\overline{f}\left(x\right)g\left(x\right)\,\text{d}x.
\end{equation}
The real and imaginary part of expressions are referred to as $\Re\left(\cdot\right)$
and $\Im\left(\cdot\right)$, respectively. The Dirac-delta distribution
on $M$ is denoted $\delta\left(\cdot\right)$ and satisfies 
\begin{equation}
f\left(y\right)=\int_{M}\,\delta\left(y\,x^{-1}\right)\,f\left(x\right)\,\text{d}x,
\end{equation}
where $y\,x^{-1}$ denotes the group product between $y$ and $x^{-1}$. 

Throughout the paper we will use different norms on several different
spaces. We will introduce them in the text whenever we use them, but
here we summarize the notation for better overview:

$\|\cdot\|_{L^{2}}=\sqrt{\left(\cdot,\cdot\right)_{L^{2}}}$ is the
$L^{2}$ norm, $\|\cdot\|_{k,\infty}$ is the family of semi norms
with respect to which the space of smooth functions is complete, in
particular $\|\cdot\|_{\infty}$ is the supremums norm for smooth
functions, $\|f\|_{\infty}=\sup_{x\in M}\left|f\left(x\right)\right|$,
$\|\cdot\|_{\star}$ refers to the $C^{\star}$-norm, $\|\cdot\|_{\mathcal{H}}$
refers to the Hilbert space norm for whatever Hilbert space is in
question, and $\|\cdot\|_{op}=\sup_{x\in\mathcal{H}}\|\cdot x\|_{\mathcal{H}}$
is the operator norm for bounded linear operators on the Hilbert space
$\mathcal{H}$.

\section{Group Field Theory \label{sec:Group-Field-Theory}}

\subsection{Operator formulation of GFT}

Group field theory is a field theoretical description of spin networks
and simplicial geometry. It can be formulated in terms of functional
integrals \cite{Freidel:2005qe,Oriti:2006se,Oriti:2009wn,Krajewski:2012aw}
or in operator language \cite{Oriti:2013aqa}. In the latter, the
natural starting point is a Fock space spanned by creation and annihilation
operators $\varphi^{\dagger}\left(x\right),\varphi\left(x\right)$,
acting on the Fock vacuum of zero quanta $|o)$, such that 
\begin{align}
\varphi^{\dagger}\left(x\right)|o) & =|x), & \varphi\left(x\right)|o) & =0.
\end{align}
In models with a simplicial or more general topological interpretation,
like the ones related to loop quantum gravity and simplicial quantum
gravity, one requires the operators to be invariant under the right
multiplication by an arbitrary element of the group $G$, such that
for all $h\in G$ we have 
\begin{equation}
\varphi\left(x_{1,}x_{2},\cdots,x_{n}\right)=\varphi\left(x_{1}h,x_{2}h,\cdots,x_{n}h\right).\label{eq:Closure constraint}
\end{equation}
This symmetry requirement is called the closure constraint \cite{Baez:1999tk,kapovich1996symplectic,Conrady:2009px,Bianchi:2010gc,Baratin:2010nn}.
The functions that satisfy the closure constraint are called gauge
invariant functions. The canonical commutation relation (CCR) between
the fields without closure constraint is given by, 
\begin{equation}
\left[\varphi\left(x\right),\varphi^{\dagger}\left(y\right)\right]=\delta\left(x\,y^{-1}\right)\,.\label{eq: CCR-2}
\end{equation}
And for gauge invariant fields the CCR read
\begin{equation}
\left[\varphi\left(x\right),\varphi^{\dagger}\left(y\right)\right]=\int_{G}\;\prod_{j=1}^{n}\,\delta\left(x_{j}hy_{j}^{-1}\right)\,\text{d}h\,.
\end{equation}
The Fock space created by these operators can be understood as a kinematical
Hilbert space $\mathcal{H}_{\mathrm{kin}}$, formed by generic quantum
states on which no dynamics has yet been imposed. As in any background-independent
formulation of quantum gravity, one expects the quantum dynamics to
be encoded in a \textit{\emph{finite set of constraint operators}}
$C_{i}:\mathcal{H}_{\mathrm{kin}}\to\mathcal{H}_{\mathrm{kin}}$ for
$i\in\left\{ 1,\cdots,N\right\} $. Following the idea of Dirac quantization
the role of $C_{i}$ is twofold: first to select the space of physical
states formally as

\begin{equation}
\mathcal{H}_{\mathrm{phys}}=\left\{ |\psi)\in\mathcal{H}_{\mathrm{kin}}\,|\,C_{i}|\psi)=0\:\forall i\in\left\{ 1,\cdots,N\right\} \right\} ,\label{eq:Physical Hilbert space}
\end{equation}
and second, select the relevant observables $\mathcal{O}$ by 
\begin{equation}
\left[C_{i},\mathcal{O}\right]=0\qquad\forall i\in\left\{ 1,\cdots,N\right\} .
\end{equation}
Any concrete choice of such operators $C_{i}$ defines a different
GFT model. In analogy to this, the constraint operators in LQG would
be the diffeomorphism constraints and the Hamiltonian constraint,
but the GFT constraints cannot be directly interpreted as diffeomorphisms
or Hamiltonian constraints, since the degrees of freedom of the GFT
theory do not live on a continuous spacetime manifold where diffeomorphisms
would be defined and act as symmetry transformations.  

A treatment of constraint systems can be technically challenging \cite{grundling1985algebraic,grundling2000local,thiemann2003lectures}.
In particular, if the zero eigenvalue lies in the continuous part
of the spectrum of $C_{i}$ the states $|\psi)$, that satisfy $C_{i}|\psi)=0$,
are not contained in the kinematical Hilbert space, and one needs
to generalize the construction using the notion of rigged Hilbert
spaces \cite{gelfand1943}. This is already the case for finite-dimensional
systems in the presence of gauge symmetries like reparametrization
invariance, and it is an even more severe issue in continuum quantum
gravity. There, it can be partially tackled by the method of refined
algebraic quantization \cite{giulini1999generality,thiemann2003lectures},
but experience with quantum field theories tells us that we need Hilbert
spaces other than Fock to describe an infinite number of interacting
degrees of freedom \cite{sewell2014quantum}. Hence, GFT combines
both types of difficulties: a constrained system without explicit
Hamiltonian evolution, and the need to study an infinite number of
degrees of freedom.

To approach this problem and establish its rigorous operator formulation,
we use the algebraic formalism for quantum statistical mechanics in
GFT. In the following we will put the above formulation of GFT in
algebraic terms and construct Hilbert spaces with infinite particle
number as representations of the GFT algebra of observables. This
will require the definition of a Weyl algebra for GFT.

\subsection{Algebraic formulation of GFT}

The first step in the construction of an algebraic formulation is
the construction of the algebra of observables. In GFT, by convenience,
we choose this algebra to be the Weyl algebra. The later is a $C^{\star}$-algebra
that is constructed over a symplectic space of the classical theory.
For that reason we start our discussion of the algebraic construction
with a definition of the suitable symplectic space in GFT. 

\subsubsection{Symplectic space of GFT}

We begin with the space of smooth, complex valued functions on $M$
that we denote by $\mathcal{S}=\mathcal{C}^{\infty}\left(M\right)$.
Let $L_{x}:M\to M$ denote the left and $R_{x}:M\to M$ the right
multiplication on $M$ by $x\in M$ and denote the pull-back of $f\in\mathcal{S}$
by $L_{x}$ (respectively $R_{x}$) as 
\begin{align}
L_{x}^{\star}f & =f\circ L_{x} & \text{} & \left(\text{respectively}R_{x}^{\star}f=f\circ R_{x}\right).
\end{align}

\begin{lem}
\label{prop1}$\mathcal{S}$ is closed under translations; that is
for any $y\in M$ and $f\in\mathcal{S}$ the functions $L_{y}^{\star}f$
and $R_{y}^{\star}f$ are again in $\mathcal{S}$. Moreover, $L_{y}^{\star}$
and $R_{y}^{\star}$ leave the $L^{2}$-bracket, $\left(\cdot,\cdot\right)_{L^{2}}$,
invariant. 
\end{lem}

\begin{proof}
The first statement follows from smoothness of the maps $L_{x}$ and
$R_{x}$. The second statement is a direct consequence of the left
(respectively right) invariance of the Haar measure $\text{d}x$.
That is for $f,g\in\mathcal{S}$ and $y\in M$,
\begin{align*}
\left(L_{y}^{\star}f,L_{y}^{\star}g\right)_{L^{2}} & =\int_{M}\,\overline{f}\left(yx\right)g\left(yx\right)\,\text{d}x\\
 & =\int_{M}\,\overline{f}\left(x\right)g\left(x\right)\,\text{d}x\\
 & =\left(f,g\right)_{L^{2}}.
\end{align*}
And similar for $R_{y}^{\star}f$.
\end{proof}
Let $X_{i}\in\mathfrak{m}$ be a Lie algebra element of $M$, then
$X_{i}$ acts as a derivation on smooth functions such that for $f\in\mathcal{S}$
, $I\subset\mathbb{R}$ an interval containing zero and $t\in I$,
\begin{equation}
X_{i}f\left(x\right)\doteq\partial_{t}\,f\left(e^{tX_{i}}\,x\right)\vert_{t=0},\label{eq:derivation on S}
\end{equation}
where $e^{tX_{i}}$ denotes the exponential map on $M$ \cite{sugiura1971}.
\begin{lem}
$\mathcal{S}$ equipped with topology induced by the family of semi-norms
\[
\left\{ \|f\|_{k,\infty}=\|X_{1}\cdots X_{k}f\left(g\right)\|_{\infty}:X_{1},\cdots,X_{k}\in\mathfrak{m};\forall k\in\mathbb{N}\right\} ,
\]
is a complete, topological, locally convex, vector space. 
\end{lem}

\begin{proof}
See reference\footnote{In this reference the authors define the lie algebra by left invariant
vector fields as opposed to our definition as right invariant vector
fields. For that reason in the original paper eq.\eqref{eq:derivation on S}
is defined by right multiplication with the exponential map. This
small change, however, does not change the results of the paper.} \cite{sugiura1971}.
\end{proof}
When the topology of $\mathcal{S}$ will be important in our discussion
we will denote this topological space by $\mathcal{S}_{\infty}.$

Since $M$ is compact, every smooth function on it is finite integrable
and we can equip $\mathcal{S}$ with the norm-topology induced by
the norm,
\begin{equation}
\|f\|_{L^{2}}^{2}=\int_{M}\,\bar{f}\left(x\right)f\left(x\right)\,\text{d}x.
\end{equation}

\begin{lem}
\label{prop4}$\mathcal{S}$ equipped with the norm topology is not
complete and its completion is the space of square integrable functions
on $M$. 
\end{lem}

\begin{proof}
See reference \cite{sugiura1971}.
\end{proof}
To distinguish this topological space from the above, we will denote
it $\mathcal{S}_{L^{2}}$ whenever this will be necessary. 

Let $h\in G$ and $D:G\to M$ be a diagonal map such that $D_{h}\equiv D\left(h\right)=\left(h,\cdots,h\right)$.
We say $f$ satisfies the closure constraint (or $f$ is gauge invariant)
if 
\begin{align}
R_{D_{h}}^{\star}f & =f &  & \forall h\in G.
\end{align}
We denote the space of functions that satisfy the closure constraint
by $\mathcal{S}_{G}$. 
\begin{prop}
\label{prop3}$\mathcal{S}$ can be decomposed in complementary subspaces
$\mathcal{S}_{G}$ and $\mathcal{S}_{NG}$ such that 
\begin{equation}
\mathcal{S}_{\infty}=\mathcal{S}_{G}+\mathcal{S}_{NG},
\end{equation}
and $\mathcal{S}_{G}\cap\mathcal{S}_{NG}=\left\{ 0\right\} $. Where
$\mathcal{S}_{G}$ is a space of gauge invariant functions and $\mathcal{S}_{NG}$
is a space of functions that do not satisfy the closure constraint. 
\end{prop}

\begin{proof}
Let $P$ define an operator on $\mathcal{S}$ and pointwise acting
as
\[
\left(Pf\right)\left(x\right)=\int_{G}\,\left(R_{D_{h}}^{\star}f\right)\left(x\right)\,\text{d}h.
\]
$P$ is linear since it is a composition of linear operators, $R_{D_{h}}^{\star}$
and $\int_{G}\,\left(\cdot\right)\,\text{d}h$. We show that the image
of $P$ is in $\mathcal{S}_{\infty}$. By \cite[lemma 2.1]{sugiura1971}
it is enough to show that $\|Pf\|_{k,\infty}<\infty$ for any $k\in\mathbb{N}$.
For an arbitrary fixed $k$ we get
\begin{align*}
\|Pf\|_{k,\infty} & =\sup_{x\in M}\left|X_{1}\cdots X_{k}\left(Pf\right)\left(x\right)\right|\\
 & =\sup_{x\in M}\left|X_{1}\cdots X_{k}\int_{G}\,\left(R_{D_{h}}^{\star}f\right)\left(x\right)\,\text{d}h\right|.
\end{align*}
By lemma \ref{prop1} the integrand is a smooth function and can be
upper bounded by $\sup_{x\in M}\left|\left(R_{D_{h}}^{\star}f\right)\left(x\right)\right|$.
Hence, by dominant convergence theorem 
\begin{align*}
\|Pf\|_{k,\infty} & \leq\int_{G}\,\sup_{x\in M}\left|X_{1}\cdots X_{k}\,\left(R_{D_{h}}^{\star}f\right)\left(x\right)\right|\,\text{d}h
\end{align*}
For any fixed $h\in G$ we have 
\begin{align*}
X_{1}\cdots X_{k}\,\left(R_{D_{h}}^{\star}f\right)\left(x\right) & =\partial_{t_{1}}\cdots\partial_{t_{k}}\,f\left(e^{t_{1}X_{1}}\cdots e^{t_{k}X_{k}}\,x\,D_{h}\right),
\end{align*}
where all derivatives are taken at zero. Since $x\,D_{h}\in M$ it
follows that 
\[
\sup_{x\in M}\left|X_{1}\cdots X_{k}\,\left(R_{D_{h}}^{\star}f\right)\left(x\right)\right|=\sup_{x\in M}\left|X_{1}\cdots X_{k}\,f\left(x\right)\right|.
\]
and we obtain 
\[
\|Pf\|_{k,\infty}\leq\|f\|_{k,\infty}.
\]
Therefore, $P:\mathcal{S_{\infty}}\to\mathcal{S}_{\infty}$, is a
continuous linear operator on $\mathcal{S}.$

Further, by right invariance of the Haar measure it follows that $P^{2}f=Pf$.
By \cite[theorem 1.1.8]{kadison1983fundamentals} it follows that
$\mathcal{S}_{\infty}$ can be decomposed as 
\begin{align*}
\mathcal{S}_{\infty} & =\mathcal{S}_{G}+\mathcal{S}_{NG},
\end{align*}
where $\mathcal{S}_{G}=P\mathcal{S}_{\infty}$ and $\mathcal{S}_{NG}=\left(1-P\right)\mathcal{S}_{\infty}$
and $\mathcal{S}_{G}\cap\mathcal{S}_{NG}=\left\{ 0\right\} $. 
\end{proof}
\begin{lem}
$P$ is an orthogonal projector on $L^{2}\left(M,\text{d}x\right)$.
\end{lem}

\begin{proof}
$P$ is bounded on $\mathcal{S}_{L^{2}}$ since for any $f\in\mathcal{S}$
we have by right invariance of the Haar measure 
\[
\|Pf\|_{L^{2}}^{2}=\int_{M}\,\int_{G}\,\left|f\left(x\,D_{h}\right)\right|\,\text{d}h\,\text{d}x=\int_{M}\,\left|f\left(x\right)\right|\,\text{d}x=\|f\|_{L^{2}}^{2}.
\]
Let $f,g\in\mathcal{S}$. Then by Fubini and the invariance of the
Haar measure under right multiplication and inversion, we have
\begin{align*}
\left(f,Pg\right)_{L^{2}} & =\int_{M}\,\overline{f}\left(x\right)\,\left(\int_{G}\,\left(R_{D_{h}}^{\star}g\right)\left(x\right)\,\text{d}h\right)\,\text{d}x\\
 & =\int_{M}\,\left(\int_{G}\,\overline{\left(R_{D_{h}}^{\star}f\right)}\left(x\right)\,\text{d}h\right)\,g\left(x\right)\,\text{d}x\\
 & =\left(Pf,g\right)_{L^{2}}.
\end{align*}
And for $h_{1},h_{2}\in G$ we have
\begin{align*}
\left(PPf\right)\left(x\right) & =\int_{G}\int_{G}\left(R_{D\left(h_{1}\right)}^{\star}R_{D\left(h_{2}\right)}^{\star}f\right)\left(x\right)\,\text{d}h_{1}\,\text{d}h_{2}\\
 & =\int_{G}\int_{G}\left(R_{D\left(h_{1}\right)}^{\star}R_{D\left(h_{2}\right)}^{\star}f\right)\left(x\right)\,\text{d}h_{1}\,\text{d}h_{2}\\
 & =\int_{G}\int_{G}\left(\left(R_{D\left(h_{1}h_{2}\right)}\right)^{\star}f\right)\left(x\right)\,\text{d}h_{1}\,\text{d}h_{2}\\
 & =\int_{G}\left(\left(R_{D_{h}}\right)^{\star}f\right)\left(x\right)\,\text{d}h=\left(Pf\right)\left(x\right)\;.
\end{align*}
Therefore, $P$ is an orthogonal projection on the dense domain of
$L^{2}\left(M,\text{d}x\right)$ and extends uniquely to the whole
$L^{2}\left(M,\text{d}x\right)$ by continuity. 
\end{proof}
\begin{thm}
The space $\mathcal{S}_{G}=P\mathcal{S}$ is dense in $PL^{2}\left(M,\text{d}x\right)$
\textemdash{} the image of the orthogonal projection $P$ on $L^{2}\left(M,\text{d}x\right)$. 
\end{thm}

\begin{proof}
Since $PL^{2}\left(M,\text{d}x\right)$ is given by the projection
$P$, it is a closed subspace of $L^{2}\left(M,\text{d}x\right)$.
By lemma \ref{prop4} the set $PL^{2}\left(M,\text{d}x\right)\cap\mathcal{S}$
is dense in $PL^{2}\left(M,\text{d}x\right)$. Further, any $f\in PL^{2}\left(M,\text{d}x\right)\cap\mathcal{S}$
is an almost-everywhere gauge invariant function that is smooth. Define
$g=f-Pf$. Then $g$ vanishes almost everywhere and is smooth. Hence
$g$ is zero everywhere, and we get $f\in\mathcal{S}_{G}$ and $PL^{2}\left(M,\text{d}x\right)\cap\mathcal{S}\subseteq\mathcal{S}_{G}$.
The opposite inclusion, $\mathcal{S}_{G}\subseteq PL^{2}\left(M,\text{d}x\right)\cap\mathcal{S}$,
is obvious since any $f\in\mathcal{S}_{G}$ is square integrable and
$\mathcal{S}_{G}\subseteq\mathcal{S}$ by lemma \ref{prop3}. 
\end{proof}
To proceed with the construction of the symplectic space we equip
$\mathcal{S}_{\infty}$ with a symplectic form $\mathfrak{s}:\mathcal{S}\times\mathcal{S}\to\mathbb{R}$
defined for any $f,g\in\mathcal{S}$ by 
\begin{equation}
\mathfrak{s}\left(f,g\right)=\Im\left[\left(f,g\right)_{L^{2}}\right].
\end{equation}
Restricting $\mathcal{S}$ to the subspace $\mathcal{S}_{G}$ we obtain
the symplectic form on $\mathcal{S}_{G}$ that we denote by the same
symbol $\mathfrak{s}$.

The above theorem ensures that after quantization, the one particle
Hilbert space, that is given by the $L^{2}$ closure of the $\mathcal{S}_{G}$
will be that of a quantized polygon \cite{kapovich1996symplectic}.
However, the symplectic structure of our space is different from the
symplectic structure of a single polygon. 
\begin{rem}
The space $\mathcal{S}_{G}$ is not closed under right multiplication,
meaning that in general for $f\in\mathcal{S}_{G}$ and $y\in M$ not
of the diagonal form (that is $y\neq D_{h}$ for any $h\in G$), the
function $R_{y}^{\star}f$ will not be gauge invariant. To see this
we observe 
\begin{align*}
\left(R_{y}^{\star}f\right)\left(x\right) & =f\left(xy\right),
\end{align*}
which is in general not equal to 
\begin{equation}
f\left(xD_{h}y\right)=\left(R_{y}^{\star}f\right)\left(x\,D_{h}\right).
\end{equation}
For this reason we choose the definition of the Lie algebra to be
given by right invariant vector fields on $M$ (generated by left
translation) to ensure that for any $f\in\mathcal{S}_{G}$, the function
$X_{i}f$ stays in $\mathcal{S}_{G}$. 
\end{rem}

\subsubsection{The Weyl algebra of GFT}

To define the Weyl algebra from the space $\mathcal{S}$ we follow
the standard procedure presented for example in \cite{grundling1985algebraic,grundling2000local}
and that we recall below for convenience.

First we define a $\star$-algebra $\mathcal{A}\left(\mathcal{S}\right)$
such that:
\begin{enumerate}
\item The elements of $\mathcal{A\left(S\right)}$ are complex valued functions
on $\mathcal{S}$ with support consisting of a finite subset of $\mathcal{S}$.
It follows that $\mathcal{A\left(S\right)}$ is a vector space. 
\item Then we define a $\ell^{1}$ norm on $\mathcal{A\left(S\right)}$
by 
\begin{align*}
\|A\|_{1} & =\sum_{f\in\mathcal{S}}|A\left(f\right)|.
\end{align*}
The sum on the right hand side is well defined since each element
in $\mathcal{A}\left(\mathcal{S}\right)$ is supported on a finite
subset of $\mathcal{S}.$
\item For $f\in\mathcal{S}$ we define functionals $W_{\left(f\right)}\in\mathcal{A}\left(\mathcal{S}\right)$
such that for any $g\in\mathcal{S}$ 
\begin{align*}
W_{\left(f\right)}\left(g\right) & =\begin{cases}
1 & \text{if}\:f=g\,\text{pointwise}\\
0 & \text{otherwise}
\end{cases}.
\end{align*}
These functionals form a dense linear subspace in $\mathcal{A\left(S\right)}$.
\item We then define the multiplication on that subspace by
\begin{align*}
W_{\left(f\right)}\cdot W_{\left(g\right)} & =e^{-\frac{\imath}{2}\mathfrak{s}\left(f,g\right)}\,W_{\left(f+g\right)}.
\end{align*}
and extend it to the full $\mathcal{A\left(S\right)}$ by linearity. 
\item Finally, we define the involution $W_{\left(f\right)}^{\star}=W_{\left(-f\right)}$.
\end{enumerate}
Closing $\mathcal{A\left(S\right)}$ in the $\ell^{1}$ norm provides
a Banach $^{\star}$-algebra $\mathbb{A}\left(\mathcal{S}\right)$.
This algebra can be represented by bounded linear operators on some
Hilbert space. Denoting the space of all non-degenerate, irreducible
representations of $\mathbb{A}\left(\mathcal{S}\right)$ by $Irreps$,
we define the Weyl algebra.
\begin{defn}
The Weyl algebra is a $C^{\star}$-algebra over $\mathcal{S}$ obtained
by completion of $\mathbb{A}\left(\mathcal{S}\right)$ in the $C^{\star}$-norm
\begin{equation}
\|W_{\left(f\right)}\|_{\star}:=\sup_{\pi\in Irreps}\|\pi\left(W_{\left(f\right)}\right)\|_{\mathcal{H}},\label{eq:c-star norm}
\end{equation}
We denote it $\mathfrak{A}\left(\mathcal{S}\right)$.
\end{defn}

\begin{lem}
\label{lem:automorphism}For any $x\in M$ the maps $\alpha_{x}$
and $\beta_{x}$ from $\mathfrak{A}\left(\mathcal{S}\right)$ to $\mathfrak{A}\left(\mathcal{S}\right)$
defined such that for any $f\in\mathcal{S}$ 
\begin{align}
\alpha_{x}\left(W_{\left(f\right)}\right) & =W_{\left(L_{x}^{\star}f\right)}, & \beta_{x}\left(W_{\left(f\right)}\right) & =W_{\left(R_{x}^{\star}f\right)},
\end{align}
and extended to the whole $\mathfrak{A}\left(\mathcal{S}\right)$
by linearity are $\star$-automorphisms. 
\end{lem}

\begin{proof}
By definition $\alpha_{x}$ and $\beta_{x}$  are linear. Further
let $f,g\in\mathcal{S}$, then by lemma \ref{prop1}
\begin{align*}
\alpha_{x}\left(W_{\left(f\right)}W_{\left(g\right)}\right) & =\alpha_{x}\left(W_{\left(f+g\right)}e^{-\frac{\imath}{s}\Im\left(f,g\right)_{L^{2}}}\right)\\
 & =e^{-\frac{\imath}{2}\Im\left(f,g\right)_{L^{2}}}W_{\left(L_{x}^{\star}f+L_{x}^{\star}g\right)}\\
 & =e^{-\frac{i}{2}\Im\left(L_{x}^{\star}f,L_{x}^{\star}g\right)_{L^{2}}}W_{\left(L_{x}^{\star}f+L_{x}^{\star}g\right)}\\
 & =W_{\left(L_{x}^{\star}f\right)}W_{\left(L_{x}^{\star}g\right)}\\
 & =\alpha_{x}\left(W_{\left(f\right)}\right)\alpha_{x}\left(W_{\left(g\right)}\right).
\end{align*}
Also 
\begin{align*}
\alpha_{x}\left(W_{\left(f\right)}^{\star}\right) & =\alpha_{x}\left(W_{\left(-f\right)}\right)\\
 & =W_{\left(-L_{x}^{\star}f\right)}\\
 & =\left[\alpha_{x}\left(W_{\left(f\right)}\right)\right]^{\star}.
\end{align*}
We can similarly address $\beta_{x}$.
\end{proof}
Restricting $\mathcal{S}$ to $\mathcal{S}_{G}$ we obtain a subset
$\mathfrak{A}_{G}$ defined as
\begin{equation}
\mathfrak{A}_{G}=\overline{\text{span}\left\{ W_{\left(f\right)}\in\mathfrak{A}\left(\mathcal{S}\right)\vert f\in\mathcal{S}_{G}\right\} }^{\|.\|_{\mathfrak{A}\left(\mathcal{S}\right)}},
\end{equation}
where $\overline{\circ}^{\|.\|_{\mathfrak{A}\left(\mathcal{S}\right)}}$
denotes the closure in the $\mathfrak{A}\left(\mathcal{S}\right)$-$C^{\star}$-algebra
norm.
\begin{thm}
$\mathfrak{A}_{G}$ is a maximal $C^{\star}$-sub-algebra of $\mathfrak{A}\left(\mathcal{S}\right)$
that satisfies $\forall A\in\mathfrak{A}_{G}$, $\beta_{D_{h}}\left(A\right)=A$
for any $h\in G$.
\end{thm}

\begin{proof}
$\mathfrak{A}_{G}$ is spanned by Weyl elements of the form $W_{\left(f\right)}$
with $f\in\mathcal{S}_{G}\subset\mathcal{S}$, hence, $\mathfrak{A}_{G}\subset\mathfrak{A}\left(\mathcal{S}\right)$.
Since $\mathcal{S}_{G}$ is closed under addition, and multiplication
by real numbers, $\mathfrak{A}_{G}$ is closed under multiplication
and involution,
\begin{align*}
W_{\left(f\right)}W_{\left(g\right)} & =W_{\left(f+g\right)}e^{-\frac{\imath}{2}\Im\left(f,g\right)}\in\mathfrak{A}_{G},\\
W_{\left(f\right)}^{\star} & =W_{\left(-f\right)}\in\mathfrak{A}_{G}.
\end{align*}
To show that $\mathfrak{A}_{G}$ is invariant under $\beta_{D_{h}}$
for any $h\in G$ let $\left(A_{n}\right)_{n\in\mathbb{N}}$ be a
Cauchy sequence in $\mathfrak{A}_{G}$ such that 
\[
A_{n}=\sum_{i=0}^{n}c_{i}W_{\left(f_{i}\right)}\qquad\text{with}\quad c_{i}\in\mathbb{C},\quad f_{i}\in\mathcal{S}_{G}
\]
and that converges to $A\in\mathfrak{A}_{G}$. Choose $h\in G$. Then
by lemma \ref{lem:automorphism} $\beta_{D_{h}}$ is a $\star$-automorphism
on $\mathfrak{A}\left(\mathcal{S}\right)$ and the sequence $\left(\beta_{D_{h}}\left(A_{n}\right)\right)_{n\in\mathbb{N}}$
is a Cauchy sequence in $\mathfrak{A}\left(\mathcal{S}\right)$ that
converges to $\beta_{D_{h}}\left(A\right)\in\mathfrak{A}\left(\mathcal{S}\right)$.
However, if $f_{i}\in\mathcal{S}_{G}$ then $\beta_{D_{h}}\left(W_{\left(f_{i}\right)}\right)=W_{\left(R_{D_{h}}^{\star}f_{i}\right)}=W_{\left(f_{i}\right)}$
and the two sequences are identical in $\mathfrak{A}_{G}$. Thus,
the limit points have to be equal and we get, $\beta_{D_{h}}\left(A\right)=A$.
The fact that $\mathfrak{A}_{G}$ is maximal follows from proposition
\ref{prop3} and the fact that we can decompose, $\mathcal{S}=\mathcal{S}_{G}+\mathcal{S}_{NG}$
with $\mathcal{S}_{G}\cap\mathcal{S}_{NG}=\left\{ 0\right\} $. 
\end{proof}
\begin{cor}
The Weyl algebra over $\mathcal{S}_{G}$, denoted $\mathfrak{A}\left(\mathcal{S}_{G}\right)$,
is a maximal $C^{\star}$-sub-algebra of $\mathfrak{A}\left(\mathcal{S}\right)$
whose elements are invariant under $\beta_{D_{h}}$ for any $h\in G$.
\end{cor}

\begin{proof}
This follows from the fact that $\eta:\mathfrak{A}\left(\mathcal{S}\right)\to\mathfrak{A}\left(\mathcal{S}_{G}\right)$
defined on Weyl elements by 
\begin{equation}
\eta\left(W_{\left(f\right)}\right)=W_{\left(Pf\right)},
\end{equation}
and extended to $\mathfrak{A}\left(\mathcal{S}\right)$ by linearity
is an invertible $\star$-homomorphism from $\mathfrak{A}_{G}$ to
$\mathfrak{A}\left(\mathcal{S}_{G}\right)$. The later is obvious
since on $\mathfrak{A}_{G}$, $\eta$ acts as an identity.
\end{proof}
This concludes our construction of the Weyl algebra for GFT. In the
following we will not distinguish between the algebra $\mathfrak{A}\left(\mathcal{S}\right)$
and $\mathfrak{A}\left(\mathcal{S}_{G}\right)$ since all the following
statements equally apply to both cases. For that reason we will use
$\mathfrak{A}$ to refer to the Weyl algebra (gauge invariant or not)
and use $\mathcal{S}$ for the space of smooth function (gauge invariant
or not). $\mathcal{S}_{\infty}$ and $\mathcal{S}_{L^{2}}$ then refer
to the corresponding topological spaces (gauge invariant or not).
In section \ref{subsec:Explicit-representations}, however, we will
use the gauge invariant algebra $\mathfrak{A}\left(\mathcal{S}_{G}\right)$
since it is more relevant for GFT's with simplicial interpretation.

\subsection{Algebraic states\label{subsec:Algebraic-states}}

In order to deal with states directly at the level of the algebra,
we briefly introduce the concept of \textit{\emph{algebraic state}}s.
An algebraic state is a linear, positive, normalized functional on
the algebra $\mathfrak{A}$, 
\[
\omega:\mathfrak{A}\to\mathbb{C}
\]
such that for any $A\in\mathfrak{A}$ we get 
\begin{align}
\omega\left(A^{\star}A\right) & \geq0\\
\omega\left(\mathds{1}\right) & =1\;.\nonumber 
\end{align}
The first inequality is the condition of positivity and the second
is the normalization. Specifically for the Weyl algebra the positivity
condition reads as follows: 
\begin{defn}
The functional $\omega:\mathfrak{A}\to\mathbb{C}$ is positive if,
for any finite $N\in\mathbb{N}$ and any set of complex coefficients
$\left\{ c_{n}\right\} _{n\in\left\{ 0,\cdots,N\right\} }$ and test
functions $\left\{ f_{n}\in\mathcal{S}\right\} _{n\in\left\{ 0,\cdots,N\right\} }$,
the following holds 
\begin{align*}
\sum_{n,m}^{N}c_{n}\bar{c}_{m}\,\omega\left(W_{\left(f_{n}-f_{m}\right)}\right)e^{-\imath\frac{\Im\left(f_{n},f_{m}\right)}{2}} & \geq0\;.
\end{align*}
\end{defn}

By the GNS construction, every algebraic state provides a triple
$\left(\mathcal{H}_{\omega},\pi_{\omega},\vert\Omega)\right)$, where
$\mathcal{H}_{\omega}$ is a Hilbert space, $\pi_{\omega}:\mathfrak{A}\to\mathcal{L}\left(\mathcal{H}_{\omega}\right)$
is a representation of $\mathfrak{A}$ in terms of bounded linear
operators on $\mathcal{H}_{\omega}$, and the state vector $\vert\Omega)\in\mathcal{H}_{\omega}$,
such that $\forall A\in\mathfrak{A}$ 
\begin{equation}
\omega\left(A\right)=\left(\Omega\vert\pi_{\omega}\left[A\right]\vert\Omega\right)\;.\label{eq:Expectation value}
\end{equation}
This representation is unique, up to unitarily equivalence \cite{Strocchi:2012ir}.

The algebra of observables $\pi_{\omega}\left(\mathcal{\mathfrak{A}}\left(\mathcal{S}\right)\right)$
on the GNS Hilbert space $\mathcal{H}_{\omega}$ is a sub-algebra
of bounded linear operators on $\mathcal{H}_{\omega}$, that we denote
$\mathfrak{M}$. The commutant of $\mathfrak{M}$ is a subset of bounded
linear operators of $\mathcal{L\left(H_{\omega}\right)}$ on $\mathcal{H}_{\omega}$
such that 
\begin{equation}
\mathfrak{M}^{'}=\left\{ A\in\mathcal{L}\left(\mathcal{H}_{\omega}\right)\vert\,\forall B\in\mathfrak{M}\:AB=BA\right\} .
\end{equation}
Usually, $\mathfrak{M}$ is not closed in the strong operator topology
on $\mathcal{H}_{\omega}$. This is because the $C^{\star}$- norm
(equation \eqref{eq:c-star norm}) is stronger than the operator norm.
The closure of $\mathfrak{M}$ in the strong (or equivalently, weak)
operator topology is called the von Neumann algebra and is equal to
the bicommutant of $\mathfrak{M}$ by the von Neumann theorem (see
for example \cite{bratteli1997operator}). We denote the von Neumann
algebra of the $\omega$-GNS representation $\mathfrak{M}^{''}$.

The center of the von Neumann algebra is then defined as $Z=\mathfrak{M}^{'}\cap\mathfrak{M}^{''}$.
A state is called factor if the center of its von Neumann algebra
contains only multiples of identity.

A state $\omega$ is called pure if it can not be written as a convex
combination of two or more states 
\[
\omega=\lambda\omega_{1}+\left(1-\lambda\right)\omega_{2}\qquad0<\lambda<1.
\]
where $\omega_{1},\omega_{2},\omega$ are pairwise distinct. Otherwise
it is called mixed. The GNS representation of a state is irreducible
if and only if the state is pure \cite[Theorem 2.3.19]{bratteli1997operator}.
The GNS representation of a state is irreducible if the state is factor.

Most algebraic states are mere mathematical artifacts, and one needs
a prescription for selecting interesting specific states that can
be considered of physical relevance. One strategy is to rely on the
quantum dynamics, encoded in a constraint operator. From the algebraic
point of view the constraint operator is therefore related to the
choice of the folium, or conversely, information about the constraint
operator is partly encoded in the algebraic state.

We will not discuss the constraint operator explicitly, since little
is known at present about the constraint operators underlying specific
GFT models. Instead and reasonably, using the following criteria starting
from the Fock representation of GFT, we consider state sequences that
satisfy two conditions:
\begin{enumerate}
\item All states in the sequence are coherent states.

This is mainly motivated by the use of GFT coherent states in the
extraction of an effective continuum dynamics in the series of works
\cite{Gielen:2013kla,Oriti:2016qtz,Oriti:2016ueo,Gielen:2016dss}.
Of course, coherent states are also key for the classical approximation
of any QFT, and routinely used in particle physics, many-body systems
and condensed matter theory, which provides further motivation. 
\item The particle number of the limit state diverges.

As described above, it is reasonable to expect that quantum states
that describe smooth geometries contain infinitely many particles.
This is only possible if the particle number operator in the corresponding
representation is formally divergent and by consequence, if the corresponding
representation is non-Fock. 
\end{enumerate}
In the next section we provide simple explicit examples for GFT representations
that satisfy these two properties.

\section{States and representations\label{sec:States-and-representations}}

\subsection{Fock states and the Fock representation}

The Weyl algebra $\mathfrak{A}$ admits the Fock representation, which
is given by the GNS representation of the algebraic state 
\begin{equation}
\omega_{F}\left(W_{\left(f\right)}\right)=e^{-\frac{\|f\|_{L^{2}}^{2}}{4}}\;.\label{eq:Fock state}
\end{equation}
Since the above state is regular, i.e. the function $\Omega\left(t\right):=\omega_{F}\left(W_{\left(tf\right)}\right)$
for $t\in I\subset\mathbb{R}$ and any fixed $f\in\mathcal{S}$ is
smooth, the generator of the Weyl operator exists by Stone's theorem
\cite{reed1980methods}. Denoting the corresponding GNS triple by
$\left(\mathcal{H}_{F},\pi_{F},\vert o)\right)$, we can write 
\begin{equation}
(o|\pi_{F}\left[W_{\left(f\right)}\right]|o)=(o|e^{\imath\Phi_{F}\left(f\right)}|o)\;,
\end{equation}
where $\Phi_{F}\left(f\right)$ is an essentially self-adjoint generator
of $\pi_{F}\left[W_{\left(f\right)}\right]$ in the Fock representation,
defined on the dense domain 
\begin{align*}
D\left(\Phi_{F}\right) & =\left\{ \sum_{i}^{N}c_{i}\,\pi_{F}\left[W_{\left(f_{i}\right)}\right]|o)|\,c_{i}\in\mathbb{C}\,,f_{i}\in\mathcal{S},\,N\in\mathbb{N}\right\} .
\end{align*}
We can obtain the action of $\Phi_{F}\left(f\right)$ on $D\left(\Phi_{F}\right)$
by differentiation. For any $\vert\psi)\in D\left(\Phi_{F}\right)$
and appropriate set of complex coefficients $\left\{ c_{i}\right\} _{i\in\left\{ 0,\cdots,N\right\} }$
and test functions $\left\{ f_{i}\right\} _{i\in\left\{ 0,\cdots,N\right\} }$
such that 
\begin{equation}
\vert\psi)=\sum_{i=0}^{N}c_{i}\,\pi_{F}\left[W_{\left(f_{i}\right)}\right]\vert o),
\end{equation}
we get
\begin{equation}
(o\vert\Phi_{F}\left(f\right)\vert\psi)=-\imath\partial_{t}\,\omega_{F}\left(W_{\left(tf\right)}\sum_{i=0}^{N}c_{i}\,W_{\left(f_{i}\right)}\right)\vert_{t=0}.\label{eq:Phi generator}
\end{equation}
In particular we obtain for any $f\in\mathcal{S}$
\begin{equation}
(o\vert\Phi_{F}\left(f\right)\vert o)=0,
\end{equation}
and 
\begin{equation}
\|\Phi_{F}\left(f\right)\vert o)\|_{\mathcal{H}}^{2}=(o\vert\Phi_{F}\left(f\right)\Phi_{F}\left(f\right)\vert o)=\frac{1}{2}\|f\|_{L^{2}}^{2}.\label{eq:square of Phi_F}
\end{equation}
By similar calculations it follows that the operators $\Phi_{F}\left(f\right)$
satisfy the commutation relation, for any $f,g\in\mathcal{S}$ 
\begin{equation}
\left[\Phi_{F}\left(f\right),\Phi_{F}\left(g\right)\right]=\imath\Im\left[\left(f,g\right)_{L^{2}}\right]\label{eq:CCR for field operators}
\end{equation}
We call $\Phi_{F}\left(f\right)$ the field operator of GFT. 

We can also define the creation and annihilation operators by 
\begin{align}
\psi_{F}\left(f\right) & =\frac{1}{\sqrt{2}}\left[\Phi_{F}\left(f\right)+\imath\Phi_{F}\left(\imath f\right)\right]\label{eq:anihilation operator}\\
\psi_{F}^{\dagger}\left(f\right) & =\frac{1}{\sqrt{2}}\left[\Phi_{F}\left(f\right)-\imath\Phi_{F}\left(\imath f\right)\right],
\end{align}
with $\psi_{F}\left(f\right)^{\dagger}=\psi_{F}^{\dagger}$$\left(f\right)$,
such that $\psi_{F}\left(f\right)$ is anti-linear in $f$, $\psi_{F}^{\dagger}\left(f\right)$
is linear in $f$, both are closed on $D\left(\Phi_{F}\right)$ and
fulfill the canonical commutation relations 
\begin{equation}
\left[\psi_{F}\left(f\right),\psi_{F}\left(g\right)\right]=\left[\psi_{F}^{\dagger}\left(f\right),\psi_{F}^{\dagger}\left(g\right)\right]=0
\end{equation}
and 
\begin{equation}
\left[\psi_{F}\left(f\right),\psi_{F}^{\dagger}\left(g\right)\right]=\left(f,g\right)\;.
\end{equation}
From equations \eqref{eq:anihilation operator}, \eqref{eq:Phi generator}
and \eqref{eq:square of Phi_F} it follows that
\begin{align*}
\|\psi_{F}\left(f\right)\vert o)\|_{\mathcal{H}}^{2} & =(o\vert\psi_{F}^{\dagger}\left(f\right)\psi_{F}\left(f\right)\vert o)=0,
\end{align*}
and therefore 
\begin{equation}
\psi_{F}\left(f\right)\vert o)=0,
\end{equation}
for all $f\in\mathcal{S}$. Hence, $\vert o)$ is the Fock vacuum
with respect to the annihilation operator $\psi_{F}\left(f\right)$
and the space $\mathcal{H}_{F}$ is spanned by polynomials of creation
operators $\psi_{F}^{\dagger}\left(f\right)$ applied on $\vert o)$.

The Fock state is pure and hence the GNS representation of $\omega_{F}$
is irreducible \cite{strocchi2005symmetry}. Also the Fock representation
is the unique representation (up to unitary equivalence)\textcolor{blue}{{}
}in which the particle number operator $N$ exists, formally given
by
\begin{align}
N & =\sum_{i\in\mathbb{N}}\psi_{F}^{\dagger}\left(f_{i}\right)\psi_{F}\left(f_{i}\right)\label{eq:particle number operator}
\end{align}
for some complete orthonormal basis $\left\{ f_{i}\right\} _{i\in\mathbb{N}}$
of $L^{2}\left(M,\text{d}x\right)$. 

\subsection{Coherent states and non-Fock representations}

Usually coherent states are characterized as eigenstates of the annihilation
operators in the Fock representation, and hence require a notion of
the Hilbert space for their very definition. In the algebraic approach,
this characterization is avoided by introducing a generalized notion
of coherent states directly at the level of the algebra. This is described
in \cite{honegger1990general,honegger1990general1}. Below we briefly
summarize some of the results of that work that will be important
for our discussion. 
\begin{defn}
Let $\Gamma:\mathcal{S}_{\infty}\to\mathbb{C}$ be a continuous linear
form on $\mathcal{S}_{\infty}$. A state $\omega$ defined on the
Weyl elements as
\begin{equation}
\omega_{\Gamma}\left(W_{\left(f\right)}\right)=\omega_{F}\left(W_{\left(f\right)}\right)\,e^{\imath\,\sqrt{2}\Re\left[\Gamma\left(f\right)\right]},\label{eq:coherent state}
\end{equation}
and extended to $\mathfrak{A}\left(\mathcal{S}\right)$ by linearity,
is called a coherent state. It is pure and regular \cite{honegger1990general}. 
\end{defn}

With this definition the Fock state is the special case of the above
family of coherent states for $\Gamma=0$.

Any linear functional $\Gamma$ corresponds to a well defined state
\cite{honegger1990general}. It should be noticed that there exist
even more general definitions of coherent states, but this is the
one that most closely reflects the condition of being an eigenfunction
of the annihilation operator. 
\begin{prop}[{{{\cite[Proposition 2.5]{honegger1990general1}}}}]
\label{prop:Coherent state theorem} The state $\omega$ of the above
form is equivalent to the Fock state, if and only if $\Gamma$ is
continuous on $\mathcal{S}_{L^{2}}$. 
\end{prop}

The detailed proof of this proposition is presented in \cite{honegger1990general1},
but we provide an intuitive sketch.

Assume that $\Gamma$ is a continuous functional on $\mathcal{S}_{L^{2}}$,
and hence, it extends by continuity to $L^{2}\left(M,\text{d}x\right)$.
Then by the Riesz lemma there exists an $\gamma\in L^{2}\left(M,\text{d}x\right)$
such that for any $f\in S_{L^{2}}$ 
\begin{equation}
\Gamma\left(f\right)=\int_{M}\,f\left(x\right)\cdot\bar{\gamma}\left(x\right)\,\text{d}x,
\end{equation}
and 
\begin{equation}
\|\Gamma\|_{op}=\|\gamma\|_{L^{2}}.
\end{equation}

The state $\omega_{\Gamma}$ provides a GNS triple $\left(\mathcal{H}_{\Gamma},\pi_{\Gamma},\vert\Gamma)\right)$.
It is not difficult to see that in this case the GNS Hilbert space
is Fock and that $\overline{L\left(f\right)}$ is the eigenvalue of
the state vector $|\Gamma)$ \cite{honegger1990general}, i.e. 
\begin{equation}
\psi_{\Gamma}\left(f\right)|\Gamma)=\overline{\Gamma\left(f\right)}|\Gamma)=\left(f,\gamma\right)_{L^{2}}|\Gamma)\;.\label{eq:Eigenvalue equation}
\end{equation}
Since the representation is Fock, the particle number operator, eq.
\eqref{eq:particle number operator}, exists and its expectation value
is given by 
\begin{equation}
(\Gamma|N|\Gamma)=\sum_{i\in\mathbb{N}}\left|\Gamma\left(f_{i}\right)\right|^{2}=\|\gamma\|=\|\Gamma\|_{op}\;.
\end{equation}
That is, the particle number is given by the $L^{2}$ norm of $\gamma$
or equivalently the operator norm of $\Gamma$. When $\Gamma$ is
discontinuous on $\mathcal{S}_{L^{2}}$ and, hence, unbounded on $L^{2}\left(M,\text{d}x\right)$
the global particle number is ill-defined and the representation can
not be Fock.

The non-Fock coherent states are hence classified by functionals $\Gamma$
which are continuous on $\mathcal{S}_{\infty}$ but discontinuous
on $\mathcal{S}_{L^{2}}$, sometimes called the space of tempered
microfunctions. 

By the Riesz-Markov theorem every functional $\Gamma$ on $\mathcal{S}_{\infty}$
is of the form 
\begin{equation}
\Gamma\left(f\right)=\int_{M}\:f\left(x\right)\,\text{d}\nu\left(x\right),\label{eq:Riesz decomposition theorem}
\end{equation}
for some Baire measure $\nu$.

From this we can easily state 
\begin{cor}
\label{cor:Invariance -> Fock}If $\Gamma$ is invariant under left
multiplication i.e. for any $y\in M$, $\Gamma\left(L_{y}^{\star}f\right)=\Gamma\left(f\right)$
for any $f\in\mathcal{S}$ , then the coherent state $\omega_{\Gamma}$
is Fock. 
\end{cor}

\begin{proof}
Let $\Gamma$ be invariant under left translations. Then for any $f\in\mathcal{S}$
we have 
\begin{equation}
\Gamma\left(L_{y}^{\star}f\right)=\int_{M}\:L_{y}^{\star}f\left(x\right)\,\text{d}\nu\left(x\right)=\Gamma\left(f\right)=\int_{M}\:f\left(x\right)\,\text{d}\nu\left(x\right),
\end{equation}
hence the measure $\nu$ is a left-invariant measure on $M$. By uniqueness
of the Haar measure, $\text{d}\nu=c\cdot\text{d}x,$ for some $c\in\mathbb{R}$.
Then by Hölder's inequality $\left|\Gamma\left(f\right)\right|\leq c\|f\|_{L^{2}}$,
and hence $\Gamma$ is continuous on $L^{2}\left(M,\text{d}x\right)$
. 
\end{proof}

\subsubsection{Remarks on the discontinuity of $\Gamma$}

From the above discussion it follows that in order to have inequivalent
coherent state representations we need the integrand in equation \eqref{eq:Riesz decomposition theorem}
to diverge on some square integrable functions on $M$. There are
two reasons for which the functional in equation \eqref{eq:Riesz decomposition theorem}
can become unbounded on $L^{2}\left(M,\text{d}x\right)$, which are
related to the long (IR) and short (UV) scale behavior of the measure
$\text{d}v$.

The IR divergences appear when the integral becomes infinite due to
regions with arbitrary large measure. This is what happens in ordinary
many-body physics. On a compact manifold, however, IR divergences
can not occur. But the UV divergence can. 

Physically, an IR divergent state can be understood as a state with
an infinite number of quanta but with a finite density. On finite
regions of the base manifold the particle number is, however, finite.
This is the typical situation in condensed matter physics \cite{emch2009algebraic}.
A UV divergence, on the other hand, corresponds to a state in which
infinitely many particles are concentrated at a single point on the
base manifold and, correspondingly, the density at this point blows
up. The particle number operator is defined globally except for such
a local region with infinite density. From the point of view of field
theory on spacetime, this situation is clearly not physical: an infinite
number of particles in a finite region corresponds to an infinite
energy density. Accordingly, quantum field theories on compact spacetimes
require a finite particle number and hence forces us to stay in the
Fock representation. This requirement is usually captured in the statement
that no phase transition can occur in field theories in a finite volume
(for example \cite{sewell2014quantum}).

In GFTs, however, the notion of energy is not present and the base
manifold does not relate to local regions of space-time. Thus, even
in the compact case, the restriction to the Fock representation would
not be well-motivated. In fact, UV divergences in the above sense
could even be desirable from the point of view of the interpretation
of GFT quanta as ``building blocks of spacetime and geometry.''
Heuristically, these types of coherent states would correspond to
condensates with a collective wave-function sharply peaked on a given
value of the underlying discrete connection. Wave functions of this
type have been used for condensate states more general than coherent
states, in \cite{Oriti:2015rwa,Oriti:2015qva}, while hints of similar
divergences of the GFT particle number were found in the GFT condensate
cosmology context in \cite{Pithis:2016wzf}. 

To summarize: GFT models on the compact manifold can exhibit inequivalent
representations due to UV divergences, even though the IR divergences
can not occur. 
\begin{rem}
A fundamental difference between UV and IR divergences is their behavior
under translations. Whereas the IR divergence can be generated by
translation invariant measures as in the example of the Bose-Einstein
condensation, the UV divergences on the compact manifold cannot, by
corollary \ref{cor:Invariance -> Fock}.
\end{rem}

\subsection{Example\label{subsec:Example}}

Our procedure to construct inequivalent representations is fairly
straightforward. By the above discussion, we simply need to construct
a sequence of continuous functionals $\Gamma_{n}$ on $\mathcal{S}_{\infty}$
that converge pointwise to a functional $\Gamma_{\infty}$ unbounded
on $L^{2}\left(M,\text{d}x\right)$. Here we provide a very simple
example in which the sequence of regular measures converges to a pure
point measure. It should be clear, however, that any measure that
satisfies the property of being unbounded on $L^{2}\left(M,\text{d}x\right)$
leads to a new inequivalent representation.

Let us first define the Dirac measure $\nu_{D}$, such that for any
open $U\subset M$ and $\mathds{1}\in M$ denoting the identity on
$M$,
\begin{equation}
\nu_{D}\left(U\right)=\begin{cases}
1 & \text{if }\mathds{1}\in U\\
0 & \text{otherwise}
\end{cases}.
\end{equation}
It follows that on smooth functions $f\in\mathcal{S}$ we have, 
\begin{equation}
\nu_{D}\left(f\right)=f\left(\mathds{1}\right).
\end{equation}
Such a Riesz functional is continuous on $\mathcal{S}_{\infty}$,
since
\begin{equation}
\left|\nu_{D}\left(f\right)\right|=\left|f\left(\mathds{1}\right)\right|\leq\|f\|_{\infty}.
\end{equation}
However, it is unbounded on $L^{2}\left(M,\text{d}x\right)$ due to
the possible singular behavior of functions at sets of Haar measure
zero.

Assume further a contracting sequence of open sets $\left\{ U_{n}\right\} _{n\in\mathbb{N}}$
around the identity $\mathds{1}\in M$, such that $U_{n+1}\subset U_{n}$
and $\cap_{n\in\mathbb{N}}U_{n}=\left\{ \mathds{1}\right\} $, and
consider a sequence of measures defined as 
\begin{equation}
\text{d}\nu_{n}=\frac{\chi_{U_{n}}}{\left|U_{n}\right|}\text{d}x,\label{eq:Sequence of measures}
\end{equation}
where $\chi_{U_{n}}$ is the characteristic function on $U_{n}$,
\begin{equation}
\chi_{U_{n}}\left(x\right)=\begin{cases}
1 & \text{if }x\in U_{n}\\
0 & \text{otherwise}
\end{cases},
\end{equation}
and $\left|U_{n}\right|=\int_{U_{n}}\,\text{d}x$. 
\begin{lem}
On $\mathcal{S}$ the sequence of functionals defined by (\ref{eq:Sequence of measures})
converges to the Dirac measure in the distributional sense. That is
for any $f\in\mathcal{S}$ 
\begin{equation}
\lim_{n\to\infty}\nu_{n}\left(f\right)=\nu_{D}\left(f\right)=f\left(\mathds{1}\right).
\end{equation}
\end{lem}

\begin{proof}
Since $f$ is continuous, we can find for some $\epsilon>0$ a neighborhood
$N_{\epsilon}\left(\mathds{1}\right)$ around $\mathds{1}$ on $M$
such that $\forall x\in N_{\epsilon}\left(\mathds{1}\right)$ $f\left(x\right)$
is in an $\epsilon$-ball around $f\left(\mathds{1}\right)$ in $\mathbb{C}$.
Since the sequence is contracting $\exists N\in\mathbb{N}$ such that
$\forall n>N$, $U_{n}\subset N_{\epsilon}\left(\mathds{1}\right)$
then 
\begin{align*}
\left|\nu_{n}\left(f\right)-\nu_{D}\left(f\right)\right| & =\left|\frac{1}{\left|U_{n}\right|}\int_{M}\,f\left(x\right)\,\chi_{U_{n}}\left(x\right)\text{\,d}x-f\left(\mathds{1}\right)\right|\\
 & =\frac{1}{\left|U_{n}\right|}\left|\int_{M}\,\chi_{U_{n}}\left(x\right)\left(f\left(x\right)-f\left(\mathds{1}\right)\right)\text{d}x\right|\\
 & \leq\frac{1}{\left|U_{n}\right|}\int_{M}\,\chi_{U_{n}}\left(x\right)\left|f\left(x\right)-f\left(\mathds{1}\right)\right|\text{d}x\\
 & \leq\epsilon.
\end{align*}
\end{proof}
At every finite $n$ the measure $\nu_{n}$ is absolutely continuous
with respect to the Haar measure and by the above proposition \ref{prop:Coherent state theorem}
every state 
\begin{equation}
\omega_{n}\left(W_{\left(f\right)}\right):=\omega_{F}\left(W_{\left(f\right)}\right)\cdot e^{\imath\sqrt{2}\Re\left[\Gamma_{n}\left(f\right)\right]},
\end{equation}
is equivalent to the Fock one. Where $\Gamma_{n}\left(f\right)\doteq\int_{M}\,f\left(x\right)\,\text{d}\nu_{n}$.
From the convergence of the measure, the convergence of the algebraic
sequence is obvious. 
\begin{lem}
The sequence of states $\omega_{n}$ converges in the $w^{\star}$-topology
to $\omega_{D}^{\mathds{1}}$ , defined on Weyl elements such that
for each $f\in\mathcal{S}$
\begin{equation}
\omega_{D}^{\mathds{1}}\left(W_{\left(f\right)}\right)\doteq\omega_{F}\left(W_{\left(f\right)}\right)\cdot e^{\imath\sqrt{2}\Re\left(\Gamma\left[f\right]\left(\mathds{1}\right)\right)},
\end{equation}
and extended by linearity to the whole $\mathfrak{A}$. 
\end{lem}

\begin{proof}
For any $W_{\left(f\right)}\in\mathfrak{A}$ we have 
\begin{align*}
 & \left|\omega_{n}\left(W_{\left(f\right)}\right)-\omega_{D}^{\mathds{1}}\left(W_{\left(f\right)}\right)\right|\\
= & \left|\omega_{F}\left(W_{\left(f\right)}\right)\right|\left|e^{\imath\sqrt{2}\Re\left[\int f\,\text{d}\nu_{n}\right]}-e^{\imath\sqrt{2}\Re\left[\int f\,\text{d}\nu_{D}\right]}\right|\\
= & \left|\omega_{F}\left(W_{\left(f\right)}\right)\right|\left|e^{\imath\sqrt{2}\Re\left[\int f\,\text{d}\nu_{n}-\int f\,\text{d}\nu_{D}\right]}-1\right|\\
\to & 0.
\end{align*}
By linearity of the state and the product property of the Weyl algebra,
this extends to the whole algebra $\mathfrak{A}$. 
\end{proof}
At finite $n$ the representation is Fock, the particle number operator
exists and the particle number of the $n$th member of the sequence
is given by 
\begin{equation}
\|\Gamma_{n}\|_{op}=\frac{1}{\left|U_{n}\right|}\;.
\end{equation}
But with increasing $n$ the particle number grows since the volume
of $U_{n}$ shrinks. At the limit point the total particle number
diverges and the corresponding representation becomes inequivalent
to the Fock one.

We can define states $\omega_{D}^{x}$ peaked at points $x\in M$
using the automorphisms $\alpha_{x^{-1}}$ introduced in the previous
section such that 
\begin{equation}
\omega_{D}^{x}=\omega_{D}^{\mathds{1}}\circ\alpha_{x^{-1}}\;.
\end{equation}
We will show in the next section that each of the states $\omega_{D}^{x}$
leads to an inequivalent representation and breaks translation invariance.

\subsection{Explicit representations\label{subsec:Explicit-representations}}

In this section we will focus on the algebra $\mathfrak{A}\left(\mathcal{S}_{G}\right)$,
since it is more relevant for GFT's with simplicial interpretation,
however, all the constructions can be directly applied to $\mathfrak{A}\left(\mathcal{S}\right)$
leading to similar results. 

We now construct an explicit representation that is generated by the
above algebraic state following the construction in \cite{araki1963representations}.

Take $L^{2}\left(M,\text{d}\nu_{D}^{x}\right)$ to be the space of
$L^{2}$ functions with respect to the Dirac measure concentrated
at $x\in M$, i.e. for any $f\in\mathcal{S}_{G}$
\begin{equation}
\nu_{D}^{x}\left(f\right)=f\left(x\right).
\end{equation}
The space $L^{2}\left(M,\text{d}\nu_{D}^{x}\right)$ is one-dimensional.
For any $f\in\mathcal{S}_{G}$ define commuting operators $A\left(f\right)$
and $B\left(f\right)$ on $L^{2}\left(M,\text{d}\nu_{D}^{x}\right)$
such that for any $\varphi\in L^{2}\left(M,\text{d}\nu_{D}^{x}\right)$
and $f\in\mathcal{S}_{G}$ 
\begin{align*}
\left[A_{x}\left(f\right)\varphi\right]\left(x\right) & =f\left(x\right)\varphi\left(x\right) & \left[B_{x}\left(f\right)\varphi\right]\left(x\right) & =\bar{f}\left(x\right)\varphi\left(x\right).
\end{align*}
We define the state vector 
\begin{equation}
|\Omega_{D}^{x})\equiv|o)\otimes1
\end{equation}
where $1$ is the constant function on $M$ and $|o)$ is the Fock
vacuum. Further we define unitary operators 
\begin{equation}
W_{\left(f\right)}^{x}=e^{\frac{\imath}{\sqrt{2}}\left[\psi_{F}\left(f\right)+\psi_{F}^{\dagger}\left(f\right)\right]}\otimes e^{\frac{\imath}{\sqrt{2}}\left[A\left(f\right)+B\left(f\right)\right]},
\end{equation}
where $\psi_{F}\left(f\right),\psi_{F}^{\dagger}\left(f\right)$ are
the Fock operators. We denote the closure of the space generated by
polynomials of operators $W_{\left(f\right)}^{x}$ acting on $\vert\Omega_{D}^{x})$
by $\mathcal{H}_{x}$. It follows that the operator algebra spanned
by $W_{\left(f\right)}^{x}$ for $f\in\mathcal{S}_{G}$ is equivalent
to $\mathfrak{M}$ (the $\omega_{D}^{x}$-GNS representation of the
Weyl algebra $\mathfrak{A}\left(\mathcal{S}_{G}\right)$) since the
expectation values coincide,
\begin{equation}
(\Omega_{D}^{x}|W_{\left(f\right)}^{x}|\Omega_{D}^{x})=e^{-\frac{\|f\|_{L^{2}}^{2}}{4}}\cdot e^{\imath\sqrt{2}\Re\left[f\left(x\right)\right]}.\label{eq:representation}
\end{equation}
Irreducibility and cyclicity of this representation are inherited
from the Fock representation since $PL^{2}\left(M,\nu_{D}\right)$
is one-dimensional.

Let $f_{y}$ be a real valued function on $M$ defined such that for
some fixed $a\in\mathbb{R}$
\begin{equation}
f_{y}\left(x\right)=\begin{cases}
a & \text{if}\quad\exists h\in G\,\text{such that \ensuremath{x=y\,D_{h}}}\\
0 & \text{else}
\end{cases},
\end{equation}
clearly $f_{y}\in PL^{2}\left(M,\text{d}x\right)$ and is zero almost
everywhere with respect to the Haar measure.
\begin{lem}
\label{lem:center of rep}Let $\left\{ f_{n}\vert f_{n}\in\mathcal{S}_{G}\right\} _{n\in\mathbb{N}}$
be a sequence that converges to $f_{y}$ in the $L^{2}$-norm. Then
the limit $\lim_{n\to\infty}W_{\left(f_{n}\right)}^{x}$ exists in
$\mathfrak{M}$. We call this element $W_{\left(f_{y}\right)}^{x}$.
Moreover, $W_{\left(f_{y}\right)}^{x}$ is in the center and there
exists a complex number $c\in\mathbb{C}$ such that $W_{\left(f_{y}\right)}^{x}=c\mathds{1}$.
\end{lem}

\begin{proof}
Since $f_{y}\in PL^{2}\left(M,\text{d}x\right)$ and $\mathcal{S}_{G}$
is dense in $PL^{2}\left(M,\text{d}x\right)$ there exists a Cauchy
sequence $\left\{ f_{n}\vert f\in\mathcal{S}_{G}\right\} $ that converges
to $f_{y}$. Then for $n$ and $m$ large enough and any $g\in\mathcal{S}_{G}$
we get by direct calculation
\begin{align*}
 & \|\left(W_{\left(f_{n}\right)}^{x}-W_{\left(f_{m}\right)}^{x}\right)W_{\left(g\right)}^{x}\vert\Omega_{D}^{x})\|_{\mathcal{H}}\\
 & =2(\Omega_{D}^{x}\vert\Omega_{D}^{x})\\
 & -2\Re\left[e^{-\frac{\|f_{n}-f_{m}\|_{L^{2}}^{2}}{4}}e^{\imath\sqrt{2}\Re\left[f_{n}\left(x\right)-f_{m}\left(x\right)\right]}\right]\\
 & \times\Re\left[e^{-\frac{\imath}{2}\Im\left[\left(f_{n}-f_{m},g\right)_{L^{2}}+\left(-f_{m}-g,f_{n}+g\right)_{L^{2}}\right]}\right]\\
 & \leq\epsilon.
\end{align*}
Since $\vert\Omega_{D}^{x})$ is cyclic we can reach every element
of $\mathcal{H}_{x}$ acting on it by polynomials of Weyl operators.
Hence, the sequence $\left\{ W_{\left(f_{n}\right)}^{x}\vert f_{n}\in\mathcal{S}_{G}\right\} $
is a Cauchy sequence in the strong operator topology and therefore
converges to an element in the von Neumann algebra $\mathfrak{M}^{''}$.
We call this element $W_{\left(f_{y}\right)}^{x}$. For any $f\in\mathcal{S}_{G}$
we have 
\begin{align*}
W_{\left(f\right)}^{x}W_{\left(f_{y}\right)}^{x} & =\lim_{n\to\infty}W_{\left(f+f_{n}\right)}^{x}e^{-\frac{\imath}{2}\Im\left[\left(f,f_{n}\right)\right]}\\
 & =W_{\left(f+f_{y}\right)}^{x}\\
 & =W_{\left(f_{y}\right)}^{x}W_{\left(f\right)}^{x},
\end{align*}
where the second equality follows from the fact that $W_{\left(f+f_{n}\right)}^{x}$
is a Cauchy sequence and $f_{y}$ is zero almost everywhere. Hence,
the element $W_{\left(f_{y}\right)}^{x}$ is in the center of the
von Neumann algebra $\mathfrak{M}^{''}$. Since the state $\omega_{D}^{x}$
is pure the center contains only multiples of identity, thus there
exists a $c\in\mathbb{C}$ such that $W_{\left(f_{y}\right)}=c\mathds{1}$.
\end{proof}

\subsubsection*{Breaking of translation symmetry}

We show that for any $x\in M$ the state $\omega_{D}^{x}$ breaks
translation symmetry in the sense that for any non-trivial $y\in M$
the translation automorphism $\alpha_{y}$ can not be represented
by a unitary operator on $\mathcal{H}_{x}$. 
\begin{cor}
Let $x,y\in M$ with $y\neq\mathds{1}$, then the states $\omega_{D}^{x}$
and $\omega_{D}^{yx}$ are inequivalent. 
\end{cor}

\begin{proof}
$\omega_{D}^{x}$ and $\omega_{D}^{yx}$ are pure states and therefore
factor. By \cite[Proposition 2.4.27]{bratteli1997operator} factor
states $\omega_{D}^{x}$ and $\omega_{D}^{yx}$ are (quasi)-equivalent
if and only if the state $\omega=\frac{1}{2}\left(\omega_{D}^{x}+\omega_{D}^{yx}\right)$
is factor as well. Therefore it is enough to show that the center
of the von Neumann algebra of $\omega$ is non-trivial. For some fixed
$a\neq0\in\mathbb{R}$ define the function
\begin{equation}
f_{yx}\left(z\right)=\begin{cases}
a & \text{if }\,\exists h\in G\quad M\ni z=yx\,D_{h}\\
0 & \text{else}
\end{cases}.
\end{equation}
The GNS triple of $\omega$ is given by 
\begin{align*}
\left(\mathcal{H}=\mathcal{H}_{x}\oplus\mathcal{H}_{yx},\pi_{\omega}=\pi_{\omega_{D}^{x}}\oplus\pi_{\omega_{D}^{yx}},\vert\Omega_{D})=\vert\Omega_{D}^{x})\oplus\vert\Omega_{D}^{yx})\right),
\end{align*}
By lemma \eqref{lem:center of rep} there exists an element $W_{\left(f_{xy}\right)}$
in the von Neumann algebra of $\omega$ as a limit of an appropriate
sequence of operators $W_{\left(f_{n}\right)}$, and $W_{\left(f_{yx}\right)}$
is in the center of the von Neumann algebra. However, a direct calculation
shows that 
\begin{align*}
 & (\Omega_{D}\vert W_{\left(f_{xy}\right)}\vert\Omega_{D})\\
 & =\frac{1}{2}\left[(\Omega_{D}^{x}\vert W_{\left(f_{yx}\right)}^{x}\vert\Omega_{D}^{x})+(\Omega_{D}^{yx}\vert W_{\left(f_{yx}\right)}^{yx}\vert\Omega_{D}^{yx})\right]\\
 & =\frac{1}{2}\left(1+e^{\imath\sqrt{2}a}\right).
\end{align*}
Hence, $W_{\left(f_{y}\right)}\neq\mathds{1}$ and the GNS representation
of $\omega$ is not factor. By theorem \cite[Proposition 2.4.27]{bratteli1997operator}
the states $\omega_{D}^{x}$ and $\omega_{D}^{yx}$ are inequivalent. 
\end{proof}
Since $\omega_{D}^{x}$ and $\omega_{D}^{yx}$ are inequivalent, the
translation automorphism $\alpha_{y}$ can not be implemented by an
unitary operator for any non-trivial $y\in M$. Hence, the translation
symmetry is broken and moreover for $x,z\in M$ not of the diagonal
form the states $\omega^{x}$ and $\omega^{z}$, lead to inequivalent
representations since they are related by translation $y=zx^{-1}\in M$.

Notice that the automorphism $\alpha_{x}$ implements the isometry
of the base manifold and hence the above representations break the
isometry transformation. This is rather different from ordinary field
theory, in which Poincaré symmetry is not allowed to be broken \cite{araki1963representations,Haag:1964aa,strocchi2005symmetry,Strocchi:2012ir}.
Again, this is possible because no spacetime interpretation is attached
to the GFT base manifold.

\section{Interpretation of new representations}

Let us pontificate on the interpretation of the newly found non-Fock
representations, expanding on some of the points above.

The state $\omega_{D}^{x}$ contain infinitely many GFT quanta carrying
a label (or equivalently have the property) $x$. It is instructive
to think about the label $x$ as one of the ``continuous modes''
of the theory. Let us call this mode the \lq basic mode\rq. In this
case the representation described above is very similar to the usual
case of Bose-Einstein condensation \cite{araki1963representations}.
The creation and annihilation operators of particles in the basic
mode $x$ are given by $A$ and $B$ operators respectively. They
commute since the number of particles in this mode is infinite, which
is the manifestation of the usual Bogoliubov argument (for example
\cite{araki1963representations,strocchi2005symmetry}). States of
the Hilbert space are then created by excitations of other ``modes''
on top of the basic one and hence can be considered quantum fluctuations
over a background that is created by infinitely many particles in
the basic mode.

We can now have the following interpretation. If we relate the group
elements of GFT with the basic notion of holonomy/curvature, which
is well-justified at the discrete level we could think about the ground
state of new representations as a truly infinite gas of particles
that all carry the same geometrical information. The resulting continuum
geometry would be then reconstructed from such an infinite particle
state. This could be a generic geometry, since approximately equal
curvature building blocks can be used, if they also have progressively
vanishing size, to approximate any geometry, as in Regge calculus
\cite{Friedberg:1984ma}. Another possibility is that they could generate
a homogeneous background with the constant holonomy (curvature) $x$.
Choosing $x=\mathds{1}$ we would obtain a flat background on top
of which excitations are created by $\psi_{F}\left(f\right)$ and
$\psi_{F}^{\dagger}\left(f\right)$. The type of states created/annihilated
on top of such a condensate background would be formally analogous
to the fundamental spin network states or cylindrical functions that
are also found in the Fock Hilbert space of the theory. Importantly,
though, in these representations the role of the Fock creation and
annihilation operators is that of collective excitations and not of
single building blocks of quantum geometry. The origin of inequivalent
representations for different $x$'s stems from the fact that the
corresponding Hilbert spaces are created by excitations over backgrounds
with different geometry than the fully degenerate one corresponding
to the Fock vacuum. Being a specific case of the condensate state
with $\Gamma=0$, the Fock representation corresponds to the case
in which the background consists of no GFT quanta at all.

The above description provides a useful intuition, but it does not
amount yet to a compelling nor complete, physical interpretation.
In fact: 
\begin{enumerate}
\item The basic mode $x$ in our case is not selected by any physical principle
such as energy minimization, entropy maximization or the enforcement
of a specific physical symmetry. It is rather postulated by hand,
which makes the construction non-unique. In contrast to this, we recall,
the ground mode in condensed matter physics is selected as the minimum
of the Hamiltonian. A detailed analysis of the constraint operator
underlying interesting GFT models is necessary, before assigning any
physical interpretation to the above representations. 
\item The states $|\Omega_{D}^{x})$ are quantum states, whose physical
properties should be ascertained by computing expectation values of
observables with a clear macroscopic, geometric meaning. This obscures
the interpretation of the elements $x\in M$ in terms of holonomy/curvature
of the reconstructed geometry.
\item The form of the constraint operator at this moment is not fully understood,
however if it is symmetric under the described translation automorphisms
the inequivalent states for different $x$'s should be physically
indistinguishable and any association of geometrical properties to
the points $x$ in the inequivalent states $\omega^{x}$ would be
incorrect. 
\end{enumerate}

\section*{Conclusions}

We have constructed an algebraic formulation for GFT. We believe that
this formulation has potential, not only allowing us to formulate
problems in a rigorous way, but also to efficiently tackle some conceptual
and technical issues related to the problem of phase transitions and
continuum limits in this class of quantum gravity models. We have
used the algebraic formulation to construct inequivalent, non-Fock
representations of the GFT algebra of observables and studied its
operator algebras in absence of dynamics in the case when the base
manifold of the GFT is compact. In particular, we focused on coherent
state representations. We have given a partial symmetry characterization
of the non-Fock representations, and attempted a preliminary geometric
interpretation of them, leaving a more complete analysis to future
work.

For the non-compact base manifolds the analysis requires different
techniques since the closure constraint can not be imposed in the
same way as we did in this paper, since the Haar measure for non-compact
groups is not normalized. Nevertheless, we believe that for GFT's
without the closure constraint similar results regarding the construction
of the operator algebra and definition of its inequivalent representations
will hold true even for non-compact base manifolds. We leave a careful
and rigorous discussion of the non-compact case for future work.

\section*{Acknowledgments}

We want to thank Isha Kotecha, Alok Laddha, and Miguel Campiglia for
many fruitful discussions, and Andreas Pithis for drawing our attention
to some literature on non-Fock coherent states. We also want to thank
one of the referees for her/his useful comments and constructive criticism
on the first version of this paper.

 \bibliographystyle{/Users/Alexander/Physics/Gravity/References/kp}
\bibliography{/Users/Alexander/Physics/Gravity/References/References}

\begingroup\raggedright\begin{thebibliography}{67}
\expandafter\ifx\csname natexlab\endcsname\relax\def\natexlab#1{#1}\fi

\bibitem[Oriti(2009)]{Oriti:2006se}
D.~Oriti, ``{The Group field theory approach to quantum gravity}'', {\em
  Approaches to Quantum Gravity, Editor D. Oriti, Cambridge University Press,
  Cambridge}, 2009 310--331,
 \href{http://xxx.lanl.gov/abs/gr-qc/0607032}{{\ttfamily arXiv:gr-qc/0607032}}.

\bibitem[Oriti(2011)]{Oriti:2011jm}
D.~Oriti, ``{The microscopic dynamics of quantum space as a group field
  theory}'', in ``{Proceedings, Foundations of Space and Time: Reflections on
  Quantum Gravity: Cape Town, South Africa}'', pp.~257--320.
\newblock 2011.
\newblock
 \href{http://xxx.lanl.gov/abs/1110.5606}{{\ttfamily arXiv:1110.5606}}.
\newblock

\bibitem[Krajewski(2011)]{Krajewski:2012aw}
T.~Krajewski, ``{Group field theories}'', {\em Proceeding of science}
  {\bfseries QGQGS} (2011) 005,
 \href{http://xxx.lanl.gov/abs/1210.6257}{{\ttfamily arXiv:1210.6257}}.

\bibitem[Thiemann(2003)]{Thiemann:2002nj}
T.~Thiemann, ``{Lectures on loop quantum gravity}'', {\em Lecture Notes in
  Physics} {\bfseries 631} (2003) 41--135,
 \href{http://xxx.lanl.gov/abs/gr-qc/0210094}{{\ttfamily arXiv:gr-qc/0210094}}.

\bibitem[Rovelli(2008)]{Rovelli2008}
C.~Rovelli, ``Loop quantum gravity'', {\em Living Reviews in Relativity}
  {\bfseries 11} (2008), no.~1, 5.

\bibitem[Ashtekar and Pullin(2017)]{Ashtekar:2017iip}
A.~Ashtekar and J.~Pullin, ``{The Overview Chapter in Loop Quantum Gravity: The
  First 30 Years}'', 2017.

\bibitem[Bodendorfer(2016)]{Bodendorfer:2016uat}
N.~Bodendorfer, ``{An elementary introduction to loop quantum gravity}'', 2016.

\bibitem[Baez(2000)]{Baez:1999sr}
J.~C. Baez, ``{An Introduction to spin foam models of quantum gravity and BF
  theory}'', {\em Lecture Notes in Physics} {\bfseries 543} (2000) 25--94,
 \href{http://xxx.lanl.gov/abs/gr-qc/9905087}{{\ttfamily arXiv:gr-qc/9905087}}.

\bibitem[Perez(2013)]{Perez:2012wv}
A.~Perez, ``{The Spin Foam Approach to Quantum Gravity}'', {\em Living Reviews
  in Relativity} {\bfseries 16} (2013) 3,
 \href{http://xxx.lanl.gov/abs/1205.2019}{{\ttfamily arXiv:1205.2019}}.

\bibitem[Oriti(2014)]{Oriti:2014uga}
D.~Oriti, ``{Group Field Theory and Loop Quantum Gravity}'',
\newblock 2014.
\newblock
 \href{http://xxx.lanl.gov/abs/1408.7112}{{\ttfamily arXiv:1408.7112}}.
\newblock

\bibitem[Gurau and Ryan(2012)]{Gurau:2011xp}
R.~Gurau and J.~P. Ryan, ``{Colored Tensor Models - a review}'', {\em SIGMA}
  {\bfseries 8} (2012) 020,
 \href{http://xxx.lanl.gov/abs/1109.4812}{{\ttfamily arXiv:1109.4812}}.

\bibitem[Gurau(2016)]{Gurau:2016cjo}
R.~Gurau, ``{Invitation to Random Tensors}'', {\em SIGMA} {\bfseries 12} (2016)
  094,
 \href{http://xxx.lanl.gov/abs/1609.06439}{{\ttfamily arXiv:1609.06439}}.

\bibitem[Rivasseau(2016)]{Rivasseau:2016wvy}
V.~Rivasseau, ``{The Tensor Track, IV}'', in ``{Proceedings, 15th Hellenic
  School and Workshops on Elementary Particle Physics and Gravity (CORFU2015):
  Corfu, Greece, September 1-25, 2015}''.
\newblock 2016.
\newblock
 \href{http://xxx.lanl.gov/abs/1604.07860}{{\ttfamily arXiv:1604.07860}}.
\newblock

\bibitem[Rivasseau(2011)]{Rivasseau:2011hm}
V.~Rivasseau, ``{Quantum Gravity and Renormalization: The Tensor Track}'', {\em
  AIP Conference Proceedings} {\bfseries 1444} (2011) 18--29,
 \href{http://xxx.lanl.gov/abs/1112.5104}{{\ttfamily arXiv:1112.5104}}.

\bibitem[Rivasseau(2012)]{Rivasseau:2012yp}
V.~Rivasseau, ``{The Tensor Track: an Update}'', in ``{29th International
  Colloquium on Group-Theoretical Methods in Physics (GROUP 29) Tianjin, China,
  August 20-26, 2012}''.
\newblock 2012.
\newblock
 \href{http://xxx.lanl.gov/abs/1209.5284}{{\ttfamily arXiv:1209.5284}}.
\newblock

\bibitem[Rivasseau(2014)]{Rivasseau:2013uca}
V.~Rivasseau, ``{The Tensor Track, III}'', {\em Fortschritte der Physik}
  {\bfseries 62} (2014) 81--107,
 \href{http://xxx.lanl.gov/abs/1311.1461}{{\ttfamily arXiv:1311.1461}}.

\bibitem[Oriti(2014)]{Oriti:2014qoa}
D.~Oriti, ``{Non-commutative quantum geometric data in group field theories}'',
  {\em Fortschritte der Physik} {\bfseries 62} (2014) 841--854,
 \href{http://xxx.lanl.gov/abs/1405.1830}{{\ttfamily arXiv:1405.1830}}.

\bibitem[Baez and Barrett(1999)]{Baez:1999tk}
J.~C. Baez and J.~W. Barrett, ``{The Quantum tetrahedron in three-dimensions
  and four-dimensions}'', {\em Advances in Theoretical and Mathematical
  Physics} {\bfseries 3} (1999) 815--850,
 \href{http://xxx.lanl.gov/abs/gr-qc/9903060}{{\ttfamily arXiv:gr-qc/9903060}}.

\bibitem[Conrady and Freidel(2009)]{Conrady:2009px}
F.~Conrady and L.~Freidel, ``{Quantum geometry from phase space reduction}'',
  {\em Journal of Mathematical Physics} {\bfseries 50} (2009) 123510,
 \href{http://xxx.lanl.gov/abs/0902.0351}{{\ttfamily arXiv:0902.0351}}.

\bibitem[Bianchi et~al.(2011)Bianchi, Dona, and Speziale]{Bianchi:2010gc}
E.~Bianchi, P.~Dona, and S.~Speziale, ``{Polyhedra in loop quantum gravity}'',
  {\em Physical Review D} {\bfseries 83} (2011) 044035,
 \href{http://xxx.lanl.gov/abs/1009.3402}{{\ttfamily arXiv:1009.3402}}.

\bibitem[Baratin et~al.(2011)Baratin, Dittrich, Oriti, and
  Tambornino]{Baratin:2010nn}
A.~Baratin, B.~Dittrich, D.~Oriti, and J.~Tambornino, ``{Non-commutative flux
  representation for loop quantum gravity}'', {\em Classical and Quantum
  Gravity} {\bfseries 28} (2011) 175011,
 \href{http://xxx.lanl.gov/abs/1004.3450}{{\ttfamily arXiv:1004.3450}}.

\bibitem[Sewell(2014)]{sewell2014quantum}
G.~L. Sewell, ``Quantum theory of collective phenomena'', Courier Corporation,
  2014.

\bibitem[Benedetti et~al.(2015)Benedetti, Ben~Geloun, and
  Oriti]{Benedetti:2014qsa}
D.~Benedetti, J.~Ben~Geloun, and D.~Oriti, ``{Functional Renormalisation Group
  Approach for Tensorial Group Field Theory: a Rank-3 Model}'', {\em Journal of
  High Energy Physics} {\bfseries 03} (2015) 084,
 \href{http://xxx.lanl.gov/abs/1411.3180}{{\ttfamily arXiv:1411.3180}}.

\bibitem[Ben~Geloun et~al.(2016)Ben~Geloun, Martini, and Oriti]{Geloun:2016qyb}
J.~Ben~Geloun, R.~Martini, and D.~Oriti, ``{Functional Renormalisation Group
  analysis of Tensorial Group Field Theories on $\mathbb{R}^d$}'', {\em
  Physical Review D} {\bfseries 94} (2016), no.~2, 024017,
 \href{http://xxx.lanl.gov/abs/1601.08211}{{\ttfamily arXiv:1601.08211}}.

\bibitem[Ben~Geloun(2016)]{Geloun:2016bhh}
J.~Ben~Geloun, ``{Renormalizable Tensor Field Theories}'', in ``{18th
  International Congress on Mathematical Physics (ICMP2015) Santiago de Chile,
  Chile, July 27-August 1, 2015}''.
\newblock 2016.
\newblock
 \href{http://xxx.lanl.gov/abs/1601.08213}{{\ttfamily arXiv:1601.08213}}.
\newblock

\bibitem[Benedetti and Lahoche(2016)]{Benedetti:2015yaa}
D.~Benedetti and V.~Lahoche, ``{Functional Renormalization Group Approach for
  Tensorial Group Field Theory: A Rank-6 Model with Closure Constraint}'', {\em
  Classical and Quantum Gravity} {\bfseries 33} (2016), no.~9, 095003,
 \href{http://xxx.lanl.gov/abs/1508.06384}{{\ttfamily arXiv:1508.06384}}.

\bibitem[Carrozza and Lahoche(2017)]{Carrozza:2016tih}
S.~Carrozza and V.~Lahoche, ``{Asymptotic safety in three-dimensional SU(2)
  Group Field Theory: evidence in the local potential approximation}'', {\em
  Classical and Quantum Gravity} {\bfseries 34} (2017), no.~11, 115004,
 \href{http://xxx.lanl.gov/abs/1612.02452}{{\ttfamily arXiv:1612.02452}}.

\bibitem[Carrozza(2016)]{Carrozza:2016vsq}
S.~Carrozza, ``{Flowing in Group Field Theory Space: a Review}'', {\em SIGMA}
  {\bfseries 12} (2016) 070,
 \href{http://xxx.lanl.gov/abs/1603.01902}{{\ttfamily arXiv:1603.01902}}.

\bibitem[Carrozza et~al.(2017)Carrozza, Lahoche, and Oriti]{Carrozza:2017vkz}
S.~Carrozza, V.~Lahoche, and D.~Oriti, ``{Renormalizable Group Field Theory
  beyond melons: an example in rank four}'', 2017.

\bibitem[Koslowski and Sahlmann(2012)]{Koslowski:2011vn}
T.~Koslowski and H.~Sahlmann, ``{Loop quantum gravity vacuum with nondegenerate
  geometry}'', {\em SIGMA} {\bfseries 8} (2012) 026,
 \href{http://xxx.lanl.gov/abs/1109.4688}{{\ttfamily arXiv:1109.4688}}.

\bibitem[Dittrich and Geiller(2015)]{Dittrich:2014wpa}
B.~Dittrich and M.~Geiller, ``{A new vacuum for Loop Quantum Gravity}'', {\em
  Classical and Quantum Gravity} {\bfseries 32} (2015), no.~11, 112001,
 \href{http://xxx.lanl.gov/abs/1401.6441}{{\ttfamily arXiv:1401.6441}}.

\bibitem[Bahr et~al.(2015)Bahr, Dittrich, and Geiller]{Bahr:2015bra}
B.~Bahr, B.~Dittrich, and M.~Geiller, ``{A new realization of quantum
  geometry}'', 2015.

\bibitem[Delcamp and Dittrich(2017)]{Delcamp:2016dqo}
C.~Delcamp and B.~Dittrich, ``{Towards a phase diagram for spin foams}'', {\em
  Classical and Quantum Gravity} {\bfseries 34} (2017), no.~22, 225006.

\bibitem[Dittrich et~al.(2016)Dittrich, Mizera, and
  Steinhaus]{Dittrich:2014mxa}
B.~Dittrich, S.~Mizera, and S.~Steinhaus, ``{Decorated tensor network
  renormalization for lattice gauge theories and spin foam models}'', {\em New
  Journal of Physics} {\bfseries 18} (2016), no.~5, 053009,
 \href{http://xxx.lanl.gov/abs/1409.2407}{{\ttfamily arXiv:1409.2407}}.

\bibitem[Dittrich(2017)]{Dittrich:2014ala}
B.~Dittrich, {\em The Continuum Limit of Loop Quantum Gravity: A Framework for
  Solving the Theory}, ch.~Chapter 5, pp.~153--179.
\newblock World Scientific, 2017.

\bibitem[Ambjorn et~al.(2014)Ambjorn, G{\"o}rlich, Jurkiewicz, Kreienbuehl, and
  Loll]{Ambjorn:2014gsa}
J.~Ambjorn, A.~G{\"o}rlich, J.~Jurkiewicz, A.~Kreienbuehl, and R.~Loll,
  ``{Renormalization Group Flow in CDT}'', {\em Classical and Quantum Gravity}
  {\bfseries 31} (2014) 165003,
 \href{http://xxx.lanl.gov/abs/1405.4585}{{\ttfamily arXiv:1405.4585}}.

\bibitem[Delepouve and Gurau(2015)]{Delepouve:2015nia}
T.~Delepouve and R.~Gurau, ``{Phase Transition in Tensor Models}'', {\em
  Journal of High Energy Physics} {\bfseries 06} (2015) 178,
 \href{http://xxx.lanl.gov/abs/1504.05745}{{\ttfamily arXiv:1504.05745}}.

\bibitem[Bonzom(2016)]{Bonzom:2016dwy}
V.~Bonzom, ``{Large $N$ Limits in Tensor Models: Towards More Universality
  Classes of Colored Triangulations in Dimension $d\geq 2$}'', {\em SIGMA}
  {\bfseries 12} (2016) 073,
 \href{http://xxx.lanl.gov/abs/1603.03570}{{\ttfamily arXiv:1603.03570}}.

\bibitem[Ambjorn et~al.(2017)Ambjorn, Coumbe, Gizbert-Studnicki, Gorlich, and
  Jurkiewicz]{Ambjorn:2017tnl}
J.~Ambjorn, D.~Coumbe, J.~Gizbert-Studnicki, A.~Gorlich, and J.~Jurkiewicz,
  ``{New higher-order transition in causal dynamical triangulations}'', {\em
  Physical Review D} {\bfseries 95} Jun (2017) 124029,
 \href{http://xxx.lanl.gov/abs/1704.04373}{{\ttfamily arXiv:1704.04373}}.

\bibitem[Gielen et~al.(2013)Gielen, Oriti, and Sindoni]{Gielen:2013kla}
S.~Gielen, D.~Oriti, and L.~Sindoni, ``{Cosmology from Group Field Theory
  Formalism for Quantum Gravity}'', {\em Physical Review Letters} {\bfseries
  111} (2013), no.~3, 031301,
 \href{http://xxx.lanl.gov/abs/1303.3576}{{\ttfamily arXiv:1303.3576}}.

\bibitem[Oriti et~al.(2015)Oriti, Pranzetti, Ryan, and Sindoni]{Oriti:2015qva}
D.~Oriti, D.~Pranzetti, J.~P. Ryan, and L.~Sindoni, ``{Generalized quantum
  gravity condensates for homogeneous geometries and cosmology}'', {\em
  Classical and Quantum Gravity} {\bfseries 32} (2015), no.~23, 235016,
 \href{http://xxx.lanl.gov/abs/1501.00936}{{\ttfamily arXiv:1501.00936}}.

\bibitem[Oriti et~al.(2016{\natexlab{a}})Oriti, Pranzetti, and
  Sindoni]{Oriti:2015rwa}
D.~Oriti, D.~Pranzetti, and L.~Sindoni, ``{Horizon entropy from quantum gravity
  condensates}'', {\em Physical Review Letters} {\bfseries 116}
  (2016){\natexlab{a}}, no.~21, 211301,
 \href{http://xxx.lanl.gov/abs/1510.06991}{{\ttfamily arXiv:1510.06991}}.

\bibitem[Oriti et~al.(2016{\natexlab{b}})Oriti, Sindoni, and
  Wilson-Ewing]{Oriti:2016qtz}
D.~Oriti, L.~Sindoni, and E.~Wilson-Ewing, ``{Emergent Friedmann dynamics with
  a quantum bounce from quantum gravity condensates}'', {\em Classical and
  Quantum Gravity} {\bfseries 33} (2016){\natexlab{b}}, no.~22, 224001,
 \href{http://xxx.lanl.gov/abs/1602.05881}{{\ttfamily arXiv:1602.05881}}.

\bibitem[Oriti et~al.(2017)Oriti, Sindoni, and Wilson-Ewing]{Oriti:2016ueo}
D.~Oriti, L.~Sindoni, and E.~Wilson-Ewing, ``{Bouncing cosmologies from quantum
  gravity condensates}'', {\em Classical and Quantum Gravity} {\bfseries 34}
  (2017) 04,
 \href{http://xxx.lanl.gov/abs/1602.08271}{{\ttfamily arXiv:1602.08271}}.

\bibitem[Pithis et~al.(2016)Pithis, Sakellariadou, and Tomov]{Pithis:2016wzf}
A.~G.~A. Pithis, M.~Sakellariadou, and P.~Tomov, ``{Impact of nonlinear
  effective interactions on group field theory quantum gravity condensates}'',
  {\em Physical Review D} {\bfseries 94} (2016), no.~6, 064056,
 \href{http://xxx.lanl.gov/abs/1607.06662}{{\ttfamily arXiv:1607.06662}}.

\bibitem[Gielen and Sindoni(2016)]{Gielen:2016dss}
S.~Gielen and L.~Sindoni, ``{Quantum Cosmology from Group Field Theory
  Condensates: a Review}'', {\em SIGMA} {\bfseries 12} (2016) 082,
 \href{http://xxx.lanl.gov/abs/1602.08104}{{\ttfamily arXiv:1602.08104}}.

\bibitem[Freidel(2005)]{Freidel:2005qe}
L.~Freidel, ``{Group field theory: An Overview}'', {\em International Journal
  of Theoretical Physics} {\bfseries 44} (2005) 1769--1783,
 \href{http://xxx.lanl.gov/abs/hep-th/0505016}{{\ttfamily
  arXiv:hep-th/0505016}}.

\bibitem[Oriti(2009)]{Oriti:2009wn}
D.~Oriti, ``{The Group field theory approach to quantum gravity: Some recent
  results}'', {\em AIP Conference Proceedings} {\bfseries 1196} (2009)
  209--218,
 \href{http://xxx.lanl.gov/abs/0912.2441}{{\ttfamily arXiv:0912.2441}}.

\bibitem[Oriti(2016)]{Oriti:2013aqa}
D.~Oriti, ``{Group field theory as the 2nd quantization of Loop Quantum
  Gravity}'', {\em Classical and Quantum Gravity} {\bfseries 33} (2016), no.~8,
  085005,
 \href{http://xxx.lanl.gov/abs/1310.7786}{{\ttfamily arXiv:1310.7786}}.

\bibitem[Kapovich et~al.(1996)Kapovich, Millson,
  et~al.]{kapovich1996symplectic}
M.~Kapovich, J.~Millson, {\em et~al.}, ``The symplectic geometry of polygons in
  euclidean space'', {\em Journal of Differential Geometry} {\bfseries 44}
  (1996), no.~3, 479--513.

\bibitem[Grundling and Hurst(1985)]{grundling1985algebraic}
H.~B. Grundling and C.~Hurst, ``Algebraic quantization of systems with a gauge
  degeneracy'', {\em Communications in mathematical physics} {\bfseries 98}
  (1985), no.~3, 369--390.

\bibitem[Grundling and Lled{\'o}(2000)]{grundling2000local}
H.~Grundling and F.~Lled{\'o}, ``Local quantum constraints'', {\em Reviews in
  Mathematical Physics} {\bfseries 12} (2000), no.~09, 1159--1218.

\bibitem[Thiemann(2003)]{thiemann2003lectures}
T.~Thiemann, ``Lectures on loop quantum gravity'', {\em Quantum Gravity}, 2003
  219--229.

\bibitem[Gelfand and Neumark(1943)]{gelfand1943}
I.~Gelfand and M.~Neumark, ``On the imbedding of normed rings into the ring of
  operators in hilbert space'', {\em Matematiceskij sbornik} {\bfseries 54}
  (1943), no.~2, 197--217.

\bibitem[Giulini and Marolf(1999)]{giulini1999generality}
D.~Giulini and D.~Marolf, ``On the generality of refined algebraic
  quantization'', {\em Classical and Quantum Gravity} {\bfseries 16} (1999),
  no.~7, 2479.

\bibitem[Sugiura(1971)]{sugiura1971}
M.~Sugiura, ``Fourier series of smooth functions on compact lie groups'', {\em
  Osaka J. Math.} {\bfseries 8} (1971), no.~1, 33--47.

\bibitem[Kadison and Ringrose(1983)]{kadison1983fundamentals}
R.~V. Kadison and J.~R. Ringrose, ``Fundamentals of the theory of operator
  algebras volume i: Elementary theory'', Academic Press, 1983.

\bibitem[Strocchi(2012)]{Strocchi:2012ir}
F.~Strocchi, ``{Spontaneous Symmetry Breaking in Quantum Systems. A review for
  Scholarpedia}'', {\em Scholarpedia} {\bfseries 7} (2012), no.~1, 11196,
 \href{http://xxx.lanl.gov/abs/1201.5459}{{\ttfamily arXiv:1201.5459}}.

\bibitem[Bratteli and Robinson(1997)]{bratteli1997operator}
O.~Bratteli and D.~Robinson, ``Operator algebras and quantum statistical
  mechanics 1: C*- and w*-algebras. symmetry groups. decomproceeding of
  scienceition of states'', Springer-Verlag, 1997.

\bibitem[Reed and Simon(1980)]{reed1980methods}
M.~Reed and B.~Simon, ``Methods of modern mathematical physics. vol. 1.
  functional analysis'', Academic New York, 1980.

\bibitem[Strocchi(2005)]{strocchi2005symmetry}
F.~Strocchi, ``Symmetry breaking'', Springer Science \& Business Media, 2005.

\bibitem[Honegger and Rapp(1990)]{honegger1990general}
R.~Honegger and A.~Rapp, ``{General Glauber coherent states on the Weyl algebra
  and their phase integrals}'', {\em Physica A: Statistical Mechanics and its
  Applications} {\bfseries 167} (1990), no.~3, 945--961.

\bibitem[Honegger and Rieckers(1990)]{honegger1990general1}
R.~Honegger and A.~Rieckers, ``The general form of non-fock coherent boson
  states'', {\em Publications of the Research Institute for Mathematical
  Sciences} {\bfseries 26} (1990), no.~2, 397--417.

\bibitem[Emch(2009)]{emch2009algebraic}
G.~G. Emch, ``Algebraic methods in statistical mechanics and quantum field
  theory'', Courier Corporation, 2009.

\bibitem[Araki and Woods(1963)]{araki1963representations}
H.~Araki and E.~Woods, ``Representations of the canonical commutation relations
  describing a nonrelativistic infinite free bose gas'', {\em Journal of
  Mathematical Physics} {\bfseries 4} (1963), no.~5, 637--662.

\bibitem[Haag and Kastler(1964)]{Haag:1964aa}
R.~Haag and D.~Kastler, ``An algebraic approach to quantum field theory'', {\em
  Journal of Mathematical Physics} {\bfseries 5} (1964), no.~7, 848--861.

\bibitem[Friedberg and Lee(1984)]{Friedberg:1984ma}
R.~Friedberg and T.~D. Lee, ``{Derivation of Regge's Action From Einstein's
  Theory of General Relativity}'', {\em Nuclear Physics} {\bfseries B242}
  (1984) 145,
[,213(1984)].

\end{thebibliography}\endgroup

\end{document}